%% file: main.tex
\documentclass[screen,acmsmall,nonacm]{acmart}
\acmJournal{PACMPL}
\acmVolume{1}
\acmNumber{CONF} %
\acmArticle{1}
\acmYear{2018}
\acmMonth{1}
\acmDOI{} %
\startPage{1}
\setcopyright{none}
\usepackage[T5,T1]{fontenc}
\usepackage{mathrsfs}
\usepackage{mathtools}
\usepackage{calf}
\usepackage{xcolor}
\usepackage{names}
\usepackage{trimclip}
\usepackage{mathpartir}
\usepackage{braket}
\usepackage{multirow}
\usepackage[clock]{ifsym}
\usepackage{stackengine}
\usepackage[only,llbracket,rrbracket,Mapsto]{stmaryrd}
\usepackage{fontawesome}
\usepackage{main}
\usepackage{tikz}
\usetikzlibrary{shapes.arrows}
\usepackage{tikz-cd}
\usepackage[capitalize]{cleveref}
\usepackage{todonotes}
\newcommand{\eval}{\Downarrow_{\ExtOpn}}
\newcommand{\evalexp}{\Downarrow_{\ExtOpn}}
\newcommand{\evalcmd}{\Downarrow_{\ExtOpn/\kw{cmd}}}

  \title{A metalanguage for cost-aware denotational semantics}
  \author{Yue Niu}
  \orcid{0000-0003-4888-6042}
  \email{yuen@cs.cmu.edu}
  
  \author{Robert Harper}
  \email{rwh@cs.cmu.edu}
  \orcid{0000-0002-9400-2941}
  \affiliation{
  \institution{Carnegie Mellon University}
  \streetaddress{5000 Forbes Ave.}
  \city{Pittsburgh}
  \state{PA}
  \postcode{15213}
  \country{USA}
  }

\begin{document}
  \begin{abstract}
    We present two metalanguages for developing \emph{synthetic cost-aware denotational semantics}
    of programming languages. Extending the recent work of \citeauthor{niu-sterling-grodin-harper:2022} 
    on \textbf{calf}, a dependent type theory for both cost and behavioral verification, 
    we define two metalanguages, $\textbf{calf}^\star$ and $\textbf{calf}^\omega$, for studying cost-aware 
    metatheory. $\textbf{calf}^\star$ is an extension of \textbf{calf} with universes and inductive types, 
    and $\textbf{calf}^\omega$ is a an extension of $\textbf{calf}^\star$ with unbounded iteration.
    We construct denotational models of the simply-typed lambda calculus and 
    Modernized Algol, a language with first-order store and while loops, and 
    show that they satisfy a \emph{cost-aware} generalization of the classic Plotkin-type computational adequacy theorem.  
    Moreover, by developing our proofs in a synthetic language of 
    \emph{phase-separated} constructions of intension and extension, 
    our results easily \emph{restrict} to the corresponding extensional theorems.
    Consequently, our work provides a positive answer to the conjecture raised in 
    \citet{niu-sterling-grodin-harper:2022} and in light of \opcit{}'s work on algorithm analysis, 
    contributes a metalanguage for doing both cost-aware programming and verification and 
    cost-aware metatheory of programming languages. 
  \end{abstract}
  \maketitle 

  \tikzcdset{line width=rule_thickness}

  \input{intro}

  \input{calfpp}
  \input{stlc}
  \input{iteration}
  \input{algol}

  \input{models-calf}
  \input{model}

\input{cbpv-models}
  \input{conclusion}

  \section*{Acknowledgement}
    We are grateful to Jonathan Sterling for productive discussions on the topic of
    this research, and to Tristan Nguyen at AFOSR for his support.
  
    This work was supported in part by AFOSR under grants MURI FA9550-15-1-0053, 
    FA9550-19-1-0216, and FA9550-21-0009, in part by the National Science Foundation 
    under award number CCF-1901381, and by AFRL through the NDSEG fellowship.
    Any opinions, findings and conclusions or recommendations
    expressed in this material are those of the authors and do not necessarily
    reflect the views of the AFOSR, NSF, or AFRL.
  
  \bibliography{bibtex-references/refs-bibtex,bib,compiler-correctness}
  \appendix
  \input{appendix}
  \input{appendix-model}
\end{document}

%% file: intro.tex
\section{Introduction}

Denotational semantics is a well-established method for obtaining an 
\emph{equational theory} for program verification. 
Whereas the operational semantics of a programming language gives meaning to programs 
via closed, whole program computation, denotational semantics aims to assign a 
\emph{compositional} theory to \emph{open} programs amenable to equational/algebraic reasoning. 
A well-behaved denotational semantics respects the operational meaning of programs in the sense that 
the denotation of a program is invariant under evaluation. This property is known as \emph{soundness}. 
Conversely, for a denotational model to be useful, it must be conservative enough as an equational theory 
so that ``computations'' in the denotational model can be reflected in the operational semantics. 
This is known as \emph{computational adequacy}.%
\footnote{In the literature, the computational adequacy sometimes refer to the conjunction of 
soundness and adequacy as we have defined here.} 
Denotational semantics satisfying these properties have been studied for a long time, starting with 
Plotkin's work on the investigation of LCF as a programming language (PCF) \citep{PLOTKIN1977223}. 

Although the question of computational adequacy has been traditionally studied in the context of denotational semantics, 
recent work on program cost analysis in type theory has broached the possibility of viewing adequacy in the more general 
context of equational theories.
In particular, \citeauthor{niu-sterling-grodin-harper:2022} proposed a dependent type theory \calf{}
(\textbf{c}ost-\textbf{a}ware \textbf{l}ogical \textbf{f}ramework) that
provides a rich specification language supporting both behavioral and cost verification of functional programs. 
That work formalizes a myriad of case studies of the cost analysis of algorithms 
in the framework and proves the consistency of \calf{} via a model construction. 
As a type theory, \calf{} can be thought of as the semantic domain of a denotational semantics in the sense 
that it furnishes an equational theory for program analysis. 
Moreover, as a cost analysis framework, \calf{} does not stipulate a cost semantics for programs; 
instead, the users of the framework is responsible for specifying the cost model of the algorithms they define. 
This raises a natural question: how does one know if a cost model is reasonable relative to a given programming language? 
In the concluding remarks, the authors expressed the idea that the choice of a cost model with respect to 
an operational semantics may be justified by an \emph{internal} computational adequacy theorem in the style of Plotkin.

In this paper, we substantiate this idea and develop extensions of \calf{} that promote it to a metalanguage for 
\emph{synthetic cost-aware denotational semantics}. To illustrate our approach, we first define \calfpp{}, 
an extension of \calf{} with universes and inductive types, which we use to define a computationally adequate 
denotational semantics for the simply-typed lambda calculus (\STLC{}). To ramp up to a richer programming language, 
we define \calfiter{}, an extension of \calfpp{} with \emph{unbounded iteration}, in which we 
define a computationally adequate semantics for Modernized Algol (\MA{}), a dialect of Algol \citep{harper:2012:pfpl}.  
\paragraph{Cost-aware computational adequacy}
In both of the case studies we prove a generalization of the classic, \emph{extensional} Plotkin adequacy that we refer to as
\emph{cost-aware computational adequacy}. Roughly, whereas the classic adequacy theorem speaks about the
extensional content in both the operational and denotational semantics, cost-aware adequacy relates the 
operational cost with the denotational cost in addition to the extensional behavior of programs. 
An important contribution of our work is the fact that ordinary adequacy follows immediately from the 
cost-aware adequacy theorem, which is an instance of a more general principle of \calf{} 
as a \emph{synthetic language} for mediating the interaction of the cost (intension) and behavior (extension) of programs, 
a point that we shall expand on in \cref{sec:interaction,sec:modalities}. 

\paragraph{Synthetic denotational semantics}
The denotational models we define are also synthetic in a more traditional sense: type structure of the object language is 
implemented as simple compositions of the corresponding type structures in the metalanguage that do not involve complex 
analytic constructions typical of classic domain theory. 
This analytic-synthetic dichotomy is perhaps most well-studied in the context of classical (Euclidean) geometry. 
Euclid's \emph{Elements} makes use of the prototypical example of a \emph{synthetic} theory: the mathematical objects 
involved in the study of geometry such as points and lines are postulated to exist and satisfy certain axioms with no further ado, 
and the subject is developed with reference to only these assumptions. 
On the other hand, analytic geometry in the sense of Descartes \emph{constructs} geometrical objects from a more primitive notion
of space (\ie{} cartesian coordinates), from which the axioms of Euclid may be verified to hold. 

The benefit of synthetic theories are both theoretical and practical. The axioms of a synthetic theory are 
useful \emph{abstractions} that reveal the fundemantal structures and seal away irrelevant details of the mathematical objects at hand. 
This has a tangible impact on users of the theory; although a programming languages researcher may not care about 
how a fixed-point operator is implemented, they will certainly need to use the universal property of the fixed-point 
to prove theorems about programs. In the context using \calf{} as a metalanguage for cost-aware denotational semantics,
the synthetic nature of the theory is reflected in both the interpretation of the type structures and the treatment of 
the interaction of intension/extension. 

\subsection{\calf{}: a cost-aware logical framework}\label{sec:calf-intro}

In this section, we recall the key components of \calf{} as a type theory and framework for cost analysis; 
we defer to \citet{niu-sterling-grodin-harper:2022} for more details. 
We present a fragment of the signature of \calf{} in \cref{fig:calf}. 

\begin{figure}
   \begin{align*}
    \mathsf{step} &: \impl{\isof{X}{\tpc}} \mathbb{C} \to \tmc{X} \to \tmc{X}\\
    \mathsf{step}_{0} &: \impl{X,e} \mstep{0}{e} = e\\
    \mathsf{step}_{+} &: \impl{X,e,c_1,c_2}\\
    &\mstep{c_1}{\mstep{c_2}{e}} = \mstep{c_1 + c_2}{e}\\\\
    \tpv &: \jdg\\
    \mathsf{tm}^+ &: \tpv \to \jdg\\
    \mathsf{U} &: \tpc \to \tpv\\
    \mathsf{F} &: \tpv \to \tpc\\
    \mathsf{tm}^{\ominus}(X) &\coloneqq \tmv{\UU{X}}\\
    \mathsf{ret} &: (\isof{A}{\tpv}, \isof{a}{\tmv{A}}) \to \tmc{\F{A}}\\
    \mathsf{bind} &: \impl{\isof{A}{\tpv}, \isof{X}{\tpc}} \tmc{\F{A}} \to\\
    &(\tmv{A} \to \tmc{X}) \to \tmc{X}\\\\
    \ExtOpn &: \jdg\\
    \ExtOpn/{\mathsf{uni}} &: \impl{\isof{u,v}{\ExtOpn}} u = v\\\\
    \Op\mathcal{J} &\coloneqq \ExtOpn \to \mathcal{J}\\
    \mathsf{step}/{\ExtOpn} &: \impl{X, e, c} \Op(\mstep{c}{e} = e)\\
    \Op^+ &: \tpv \to \tpv\\
    \_ &: \impl{A} \tmv{\Op^+{A}} \cong \Op(\tmv{A})\\\\
     \Pi &: (\isof{A}{\tpv}, \isof{X}{\tmv{A} \to \tpc}) \to \tpc\\
    (\mathsf{ap}, \mathsf{lam}) &: \impl{A,X} \tmc{\Pi(A; X)} \cong (\isof{a}{\tmv{A}}) \to \tmc{X(a)} \\\\
    \mathsf{lam}_{\mathsf{step}} &: \impl{A,X,f,c} \lam{\mstep{c}{f}} = \mstep{c}{\lam{f}}\\
    \mathsf{bind}_{\mathsf{step}} &: \impl{A,X,e,f,c} \bind{\mstep{c}{e}}{f} = \mstep{c}{\bind{e}{f}}
    \end{align*}
  \caption{A fragment of the signature of \calf{}.}
  \label{fig:calf}
\end{figure}

\subsubsection{Dependent call-by-push-value}
\calf{} is defined as an extension of the dependent call-by-push-value calculus of \citet{pedrot-tabareau:2020}. 
Recall that the theory of call-by-push-value (CBPV) can be extracted from the Eilenberg-Moore category arising from a monad that 
encodes the computational effect. More concretely, there are two classes of types in CBPV: 
the \emph{value}/\emph{positive} types classifying values, and the \emph{computation}/\emph{negative} types classifying 
computations. Semantically, value types correspond to plain sets while computation types correspond to 
\emph{algebras} for the given monad. The type constructors $\kw{U}, \kw{F}$ bridge this stratification of values and computations and 
corresponds to a free-forgetful adjunction in the semantics. 
A computation of the type $\F{A}$ is called a \emph{free computation}, and 
$\kw{ret}$ and $\kw{bind}$ are the introduction and elimination forms of the free computations. 

\subsubsection{Cost as a computational effect}\label{sec:cost-effect}

As a theory, \calf{} is paramterized by an (ordered) monoid $(\mathbb{C}, +, \le, 0)$. 
The cost structure of programs is generated from a single computational effect 
$\kw{step} : \impl{X} \mathbb{C} \to \tmc{X} \to \tmc{X}$, which one may think operationally as 
incurring the given cost onto a computation. 

As a dependent CBPV calculus, \calf{} supports a simple equational theory for reasoning about the cost of computations. 
For instance, \citet{niu-sterling-grodin-harper:2022} defines an internal predicate 
$\kw{hasCost}_A(e, c) \coloneqq \Sigma a \mathbin{:} A.\; e = \mstep{c}{\ret{a}}$ that defines when a computation has a given cost.

\subsubsection{The interaction intension and extension}\label{sec:interaction}

A key innovation of \calf{} as a cost analysis framework is a solution to the problem of \emph{exotic programs}. 
Traditional accounts of cost structure in type theory employs the cost monad/writer monad $\mathbb{C} \times -$, 
so that a cost-aware/effectful program of type $A$ is rendered as a term of type $\mathbb{C} \times A$. 
One thinks of an effectful program in this setting as a program instrumented with a counter that returns the incurred cost. 
However, this encoding is transparent enough so that the counter is allowed to interfere with the \emph{behavior} of the program; 
such programs are called \emph{exotic} by \citet{niu-sterling-grodin-harper:2022} because one cannot extract from it 
an ordinary, cost-unaware program. 

Because the free computations and the cost effect \kw{step} are \emph{abstract}, there is no way to define such exotic programs,  
which is an internal theorem one may specify and prove in \calf{}. Semantically, the free computations may be 
implemented using the writer monad on an appropriate cost monoid, 
but it is important that this fact is not exposed in the theory. 
In order to work with cost effects in the abstract, \calf{} introduced a pair of \emph{modalities} for 
the interaction of intension and extension. 

\subsubsection{Modalities for intension and extension}\label{sec:modalities}

The first problem one encounters working in a cost-sensitive/intensional setting where the cost effect is abstract is 
function extensionality. For example, consider the merge sort and insertion sort algorithms. 
Under the usual cost model, these algorithms are most definitely distinct as far as cost is concerned. 
However, because they are both sorting algorithms, they are equal in extension/behavior, and by functional extensionality, 
they are equal! In \calf{} this contradiction may be resolved by the following observation: equality of extension/behavior 
may be analyzed in a special phase called the \emph{extensional phase} in which the cost effect is trivial. 
Technically, the extensional phase is generated by a distinguished proposition $\ExtOpn$ along with 
the axiom $\kw{step}/\ExtOpn$ (see \cref{fig:calf}); whenever we are in a context in which $\ExtOpn$ is derivable, 
$\kw{step}/\ExtOpn$ stipulates that $\kw{step}$ is trivial, and therefore we require ordinary extensional reasoning. 
Of course there are no closed terms of $\ExtOpn$, but no other structures are assumed aside from the fact it is a proposition. 

The extensional phase generates a pair of modalities for intension and extension. 
The extensional modality is defined as $\Op(A) \coloneqq \ExtOpn \to A$, which simply 
internalizes the derivability of $\ExtOpn$. Given a type $A$, one can think of $\Op A$ as 
the extensional part of $A$; 
in terms of the cost monad the unit of the extensional modality is the projection map $\mathbb{C} \times A \to A$. 
Complementary to the extensional modality is the intensional modality, which is defined as a pushout of the projections
of $A \times \ExtOpn$. It is a bit more difficult to visualize the meaning of the intensional modality, but one can imagine 
$\Cl A$ as identical to $A$ except that it is trivial inside the extensional phase, \ie{} $\Op \Cl A \cong 1$. 
A useful way to internalize this fact is the phrase ``the extension part of the intensional part is trivial''. 

In \calf{} one can use these modalities to manage the interaction of the intension and extension. 
For instance, although it is not the case that merge sort and insertion sort are equal in the 
\emph{empty} context, one can derive their equality in the extensional phase, \ie{} one has 
$\Op(\kw{mergeSort} = \kw{insSort})$. On the other hand, one may use the intensional modality to 
\emph{seal away} cost structures, which is useful in applications such as program optimization and noninterference. 

\paragraph{The phase distinction of intension and extension}

The interaction of the intension and extension in \calf{} is an instance of a more general phenomenon of 
\emph{phase distinctions} in the sense of the theory of ML modules; as explained in 
\citet{sterling-harper:2021,niu-sterling-grodin-harper:2022}, the (non)interaction of intensional structure 
with the extensional behavior of a cost-aware function is formally identical to the (non)interaction of dynamic components 
with static components in a module functor. 
Consequently, one can think of \calf{} as a synthetic language for \emph{phase distinct programming} of intension and extension, 
and a \calf{} program is said to be \textbf{phase-separated} if it exploits the interaction of the intensional and extensional modalities. 

\subsection{Cost-aware computational adequacy} \label{sec:cost-aware-adequacy}

\citeauthor{niu-sterling-grodin-harper:2022} deployed these ideas on several case studies, including Euclid's algorithm for 
the greatest common divisor, amortized analysis batched queues, and sequential and parallel sorting algorithms. 
An important feature of these analyses is that they all employed their own cost models, which follows the prevailing convention 
of algorithms research community. Although all the cost models of \opcit{} are intuitively reasonable, the authors did not provide a formal theory  
for \emph{why} certain cost models \emph{are} reasonable; however, it was conjectured that this may be achieved via a cost-aware version of Plotkin's 
adequacy theorem. 

We provide a positive answer to this conjecture.
For the following, suppose that we have defined inside \calf{} a programming language $\textbf{P}$ 
along with an evaluation relation $\Downarrow : \textbf{P} \to \Nat \to \textbf{P} \to \tpv$.  
A denotational semantics $\sem{-}$ of $\textbf{P}$ satisfies \textbf{cost-aware computational adequacy} when the following holds:
\begin{quote}
  For all closed programs of base type $\vdash_{\textbf{P}} e : \code{bool}$, if 
  $\sem{e} = \mstep{c}{\ret{b}}$ for some $b : \kw{bool}$ and $c : \Nat$, then $e \Downarrow_{\ExtOpn}^{\eta_{\Cl} c} \overline{b}$. 
\end{quote}
In the above, $\code{bool}$ is the boolean type in \textbf{P}, and $\overline{a}$ sends a \calf{} boolean to its numeral in $\code{bool}$. 
We write $\eta_{\Cl} : A \to \Cl A$ for the unit of the intensional modality, and 
the relation $\Downarrow_{\ExtOpn}$ is a \emph{phase-separated} version of the evaluation relation 
whose meaning we will explain shortly. 
Roughly, cost-aware computational adequacy states that if the denotation of a boolean program is equal to a value incurring some cost, then 
operationally the program must also evaluate to the same value with the same cost. 

\paragraph{Phase-separated evaluation}
In order to explain phase-separated evaluation, let us consider a statement of cost-aware adequacy using 
the usual operational cost semantics $\Downarrow$ and observe what goes wrong. Suppose we have a proof of the extensional phase 
$u : \ExtOpn$, and consider a closed program $e : \code{bool}$ such that 
$e = \mstep{c}{\ret{\Set{\kw{tt}, \kw{ff}}}}$. We have to show that $e \Downarrow^c \Set{\code{tt}, \code{ff}}$. 
But because we are in the extensional phase, we also have $\mstep{c}{\ret{\Set{\kw{tt}, \kw{ff}}}} = \mstep{0}{\ret{\Set{\kw{tt}, \kw{ff}}}}$. 
Therefore, we also have to show that $e \Downarrow^0 \Set{\code{tt}, \code{ff}}$! 
Now, if $c \ne 0$, we have a contradiction if adequacy holds, because the evaluation relation is deterministic:
if $e \Downarrow^c v$ holds, then it holds for unique $c$ and $v$.  

So we would like the relation $\Downarrow$ to \emph{restrict} to a cost-unaware evaluation relation in the extensional phase. 
This is the purpose of the phase-separated evaluation relation: we define a relation 
$\eval : \textbf{P} \to \Cl\Nat \to \textbf{P} \to \tpv$ whose cost component is \emph{sealed} by the intensional modality. 
We can define such a relation $\eval$ and prove that it becomes equivalent to the cost-unaware evaluation relation $e \Downarrow v$ 
in the extensional phase. Consequently, the problem above is resolved because the contradictory evaluation costs are sealed away 
by the intensional modality and invisible in the extensional phase. 

\subsubsection{Synthetic cost-aware denotational semantics}\label{sec:synth-cost-den-sem}

The contribution of our work is the development of the proceeding idea with two concrete programming languages: the 
simply-typed lambda calculus (\STLC{}) and Modernized Algol (\MA{}). 
First, we axiomatize an extension of \calf{} with universes and inductive types dubbed \calfpp{}.
We define the syntax and operational semantics of \STLC{} in \calfpp{}, 
construct a cost-aware denotational semantics of \STLC{}, and 
prove the model to be computationally adequate in the sense describe above. 
Next, we axiomatize an extension of \calfpp{} 
with unbounded iteration dubbed \calfiter{} and carry out a similar construction for 
\MA{}, a language with first-order stores and while loops. 

In both case studies we will rely on the respective metalanguages to define 
conceptually simple and \emph{synthetic} models of the object programming languages. 
Here we emphasize that our models are synthetic in two orthogonal senses. 
First, as we discussed in \cref{sec:modalities}, \calf{} is a theory for 
phase-separated constructions. In more geometric language, one can think of 
\calf{} types as families of cost/intensional structures indexed by behavioral/extensional specifications. 
The benefit of working in a language for such indexed constructions is that 
one may always \emph{project} the index of the family to obtain the ordinary, extensional content of 
an object. This is exemplified in our account of cost-aware computational adequacy: 
the classic extensional Plotkin-type adequacy theorem \textbf{follows immediately as a corollary} 
of cost-aware computational adequacy, 
a property that is \emph{not} enjoyed by prior work on sythetic denotational semantics, which we discuss more 
in \cref{subsec:related}.  

Our work is also synthetic in a more traditional sense: types and computations of the object language is 
defined via simple constructions using the corresponding structures in the metalanguage. 
For instance, while loops of \MA{} may be interpreted straightforwardly as an 
\emph{iteration} primitive of \calfiter{} satisfying the expected unfolding law
and compactness property. By isolating the essential properties necessary to develop the computational adequacy proof, 
we refine and provide a way to interpret classic accounts of adequacy in multiple metatheories. 

\subsection{Models of \calfpp{} and \calfiter{}}

To show that a synthetic theory is sensible, one must exhibit an interpretation or model of 
the theory in terms of previously understood concepts. 
Following the work of \citet{niu-sterling-grodin-harper:2022}, we prove the consistency of 
\calfpp{} and \calfiter{} by means of a model construction. 
Similar to authors of \opcit{}, we define \calfpp{} and \calfiter{} as signatures of the \emph{logical framework} of 
locally cartesian closed categories (lccc's). One can think of the language of this logical framework as an extensional dependent type 
theory with a universe of judgments closed under dependent products, dependent sums, and extensional equality. 
A model of a theory (\eg{} \calf{}) associated to a signature is given by an implementation of the constants in the 
signature in any other lccc. 

The authors of \opcit{} defines a model of \calf{} called the \emph{counting model} in an arbitrary presheaf topos equipped with a proposition representing 
the semantic extensional phase. The counting model is itself an extension of the Eilgenberg-Moore model of CBPV associated to 
the cost monad. We first extend the counting model of \calf{} with an interpretation of inductive types and universes 
based on the \emph{weaning models} of dependent CBPV \citep{pedrot-tabareau:2020} to obtain a model of \calfpp{}. 
We then modify this model to account for partiality, resulting in a model for \calfiter{}. 

\input{related}

%% file: related.tex
\subsection{Related work}
\label{subsec:related}

\subsubsection{Cost-aware denotational semantics}

Notions of cost-sensitive version of classic extensional denotational semantics has 
been studied in the context of formalization of recurrence extraction \citep{danner-licata-ramyaa:2015,kavvos-morehouse-licata-danner:2019}.
The idea of this series of works is to develop a formal framework for the 
well-known and used method of recurrence relations and extend it to work with higher-order functional programs. 
Here we give a brief summary of the approach of \citet{kavvos-morehouse-licata-danner:2019}, which 
is divided into several stages across two different languages: the programming language (CBPV) and
the syntactic recurrence language. A \emph{syntactic recurrence} is extracted from a given program, which is 
then turn into a semantic recurrence suitable for mathematical manipulation in using a denotational semantics based on
\emph{sized domains}. \opcit{} prove a bounding theorem about the extraction procedure, which roughly states that 
the extracted recurrence program produces an upperbound on the evaluation of the source program; 
\opcit{} then prove a traditional adequacy theorem for the denotational semantics with respect to the 
syntactic recurrence language, 
which in conjunction with the bounding theorem produces a sound mathematical domain for doing 
algorithm analysis \`a la recurrence relations.

The work of \citet{kavvos-morehouse-licata-danner:2019} is different from ours in several aspects. 
First, \opcit{} define a denotational semantics for a variant of PCF, while we define 
the denotational semantics of \MA{}, an imperative language with first-order store and while loops 
(and consequently no mechanism for defining arbitrary fixed-points). 
Second, the adequacy theorem of \opcit{} only speaks about cost indirectly 
through the recurrence extraction, whereas we prove a direct cost-aware adequacy theorem that 
allows one to directly reason about cost of the source program. 
Lastly, the most important difference stems from the fact that we work in a formalized metatheory 
(\calf{}) that allows for synthetic constructions as discussed in \cref{sec:modalities,sec:synth-cost-den-sem}. 
Whereas \citet{kavvos-morehouse-licata-danner:2019} work in a classical set-theoretic metatheory and use 
classic domain-theoretic constructions, we promote a more abstract approach based on an axiomatization of 
the necessary domain structures and the interaction of intension and extension. 
Moreover, because \calf{} (and its extensions) is a dependent type theory, 
one may also use it as programming language, which \citet{niu-sterling-grodin-harper:2022}
has shown to be a fruitful endeavor in the context of both the cost analysis and behavioral verification of programs. 

\subsubsection{Denotational semantics of Algol}

Denotational semantics of procedural languages with first-order store 
traces back to at least the early seminal works of \citet{scott:1970:outline,scott-strachey:1971}. 
Algol and Algol-like languages in particular were widely-studied in terms of both denotational and operational semantics, 
and we do not recall the specifics of the language here; for more details we defer to the 
definitive sources on the subject \citep{ohearn:1997,ohearn:1997-2,reynolds:1981}. 
We present a denotational model for Modernized Algol (\MA{}), a version Algol presented in \citet{harper:2012:pfpl}. 
\MA{} is an imperative programming language with first-order store and unbounded iteration, a call-by-value
operational semantics, and a categorical separation of expressions and commands. 
The denotational semantics we define for \MA{} is based on the a Kripke-world interpretation of store and in some sense 
not substantially different from the standard models of Algol-like languages. 

The main improvement of our work is the cost-aware aspect of the denotational semantics and an 
axiomatization of the properties in \calfiter{} that are necessary to prove computational adequacy. 
Moreover, as mentioned in \cref{sec:cost-aware-adequacy}, our results easily restricts to the classic adequacy theorem 
for denotational models of \MA{} extensionally. 

\subsubsection{Synthetic domain theory}\label{sec:SDT}

Ever since the pioneering work of \citeauthor{scott:1982} on the domain-theoretic semantics for 
programming languages, there has been much interest \citep{hyland:1991,fiore-rosolini:1997,fiore-plotkin:1996,reus-streicher:1999} 
in finding set-theoretic universes (in other words, topoi) that embed concrete categories of domains, 
which would furnish a rich intuitionistic/type-theoretic 
framework for defining reasoning about domain-theoretic constructions. 

Synthetic domain theory (SDT) is an elegant and powerful approach for modeling programming languages, 
but does not immediately provide a synthetic language for talking about cost-aware computation. 
From our perspective, a good way to situate our work with respect to SDT is to view categories 
(topoi to be precise) with an SDT theory
as \emph{models} for the kind of metalanguages we promote in this paper. In fact, 
we hope that by constructing models of \calf{} in categories with SDT structure, 
we can extend our results to Plotkin's PCF, thereby truly generalizing Plotkin's original adequacy result 
to a cost-sensitive setting; we shall come back to this point in \cref{sec:conclusion}. 

\subsubsection{Denotational semantics in guarded type theory}\label{sec:den-sem-guarded}

More recently, \citet{mogelberg-paviotti:2016,paviotti-mogelberg-birkedal:2015} promoted the use of 
\emph{guarded dependent type theory} (gDTT) for doing synthetic denotational semantics. 
Similar to our approach of using a type-theoretic metalanguage,
\citet{mogelberg-paviotti:2016} defines a denotational semantics for FPC in gDTT and prove it to be 
computationally adequate in the traditional sense. 
Interestingly, \opcit{} defines an intensional denotational model of FPC: 
because recursive types of FPC are defined using guarded recursive types in gDTT, 
the interpretation of terms of recursive types naturally contains ``steps'' 
engenered by the use of guarded recursion. 
As a result, \opcit{} works with a slightly nonstandard operational cost semantics 
that is defined to compute in ``lock-step'' with the unfolding of guarded recursive types in the denotational semantics. 
\opcit{} proves an intensional adequacy theorem that relates 
the steps taken by the operational semantics and the denotational semantics. 

There is a subtle difference between \opcit{} and our notion of intensional adequacy. Because 
the operational semantics of \opcit{} \emph{only} tracks the unfolding of recursive types, it does not
correspond to the natural cost semantics one obtains from the reflexive-transitive closure of the one step 
transition relation. We do not perceive this to be a fundamental limitation of guarded type theory, since 
one may insert artificial delays in both the operational and denotational semantics to obtain an adequacy theorem about 
an ordinary cost semantics. However, because guarded type theory does not have the general facilities of \calf{}
for reasoning about cost-aware programs (see \cref{sec:calf-intro}), it is not clear if this kind of 
result would be useful in that setting.

A more significant difference also stems from the fact that gDTT is not equipped with a 
synthetic language for phase-separated constructions. To obtain the ordinary extensional adequacy theorem,
\citet{mogelberg-paviotti:2016} employ an additional logical relation over the interpretation of FPC types, 
a construction that involves defining a guarded version of the coninductive delay monad and an analogue of the 
weak bisimilarity relation. 
In contrast, the extensional adequacy theorem in our setting follows \emph{immediately} from the cost-aware adequacy theorem, 
a consequence of working in a framework suitable for cost-aware metatheory. 
Note that the work did not disappear --- by isolating the theory of the interaction of intension and 
extension and verifying the resulting axioms via a model construction once and for all, we package up the work 
into a mechanism that may be applied more generally than the concrete analytic construction used by \opcit{}

\subsubsection{Compiler correctness}\label{sec:compiler}

Lastly, we outline some connections of our work to the area of compiler correctness 
\citet{patterson19:ccc, perconti14:fca, ahmed15:snapl, ahmed19:refcc, benton2010realizability}. 
In the early days of the mathematization of the study of programming languages,
the primary purpose of denotational semantics is to explain the meaning of programs 
in terms of previously established and undertood mathematical structures. 
However, as the field and synthetic methods developed,  
denotational semantics took on more of a logical character: 
models look more like \emph{translations} between different \emph{languages}, 
a view point that is expressed in \citet{Jung1996DomainsAD}. 
Consequently, one may view denotational semantics as a sort of ``compiler'' from 
the object language into the semantic domain and computational adequacy as a sort of 
compiler correctness argument. Traditionally, compiler verification is concerned with the 
\emph{functional} or \emph{extensional} correctness of the compilation process. 
We are naturally led to ask whether working in a rich, \emph{cost-aware} metalanguage for denotational semantics 
could prove useful in studying the \emph{intensional} aspects of compilation. 
We have broached the idea in this paper by way of proving a cost-aware adequacy theorem, and there are
ample opportunities to apply the ideas we developed to both new and old problems in compiler verification.

%% file: calfpp.tex
\section{\calfpp{}: extending \calf{} with universes and inductive types}

  In this section, we present an extension of \calf{} with universes and general inductive types. 
  \begin{figure}
    \begin{align*}
      \mathsf{Univ}^+ &: \tpv\\
      \mathsf{El}^+ &: \tmv{\vuniv} \to \tpv\\
      \mathsf{Univ}^{\ominus} &: \tpc\\
      \mathsf{El}^\ominus &: \tmv{\cuniv} \to \tpc\\\\ 
      W &: (A : \tpv, B : \tmv{A} \to \tpv) \to \tpv\\ 
      W/\kw{intro} &: \impl{A, B} (a : \tmv{A}) \to (\tmv{B(a)} \to \tmv{W(A, B)}) \to \tmv{W(A, B)}\\ 
      W/\kw{rec} &: \impl{A, B, (C : \tmv{W(A, B)} \to \tpv)} \\
        &((a : \tmv{A}) \to (f : \tmv{B(a)} \to \tmv{W(A,B)}) \to\\
        &\quad((b : \tmv{B(a)}) \to \tmv{C(f b)}) \to \tmv{C(W/\kw{intro}(a, f))}) \to \\
        &(w : \tmv{W(A,B)}) \to \tmv{C(w)}\\ 
      W/\kw{comp} &: \impl{A, B, C, h, a, f} W/\kw{rec}(W/\kw{intro}(a, f), h) = h(a, f, (\lambda b.\; W/\kw{rec}(h, f(b))))
    \end{align*}
  \end{figure}
  \subsection{Universes}
  Following \citet{pedrot-tabareau:2020}, we axiomatize a \emph{pair} of universes $\kw{Univ}^+, \kw{Univ}^\ominus$ 
  classifying value types and computation types respectively. 
  We do not explicitly write down the type codes and their decodings (which are completely standard); 
  as an example, the following signature axiomatizes closure under dependent products: 
  \begin{align*}
    \widehat{\Pi} &: (A : \kw{Univ}^+) \to (\kw{El}^+(A) \to \kw{Univ}^\ominus) \to \kw{Univ}^\ominus\\ 
    \widehat{\Pi}/\kw{decode} &: \impl{A, X} \kw{El}^\ominus(\widehat{\Pi}(A, X)) = \Pi(\kw{El}^+(A), \lambda a.\; \kw{El}^\ominus(X(a)))
  \end{align*}

  \paragraph{Convention}
  In this paper we define type families (\ie{} functions whose codomain is $\tpv$) in a style akin to 
  large elimination that can be unfolded to defining 
  a family of type \emph{codes} (\ie{} functions whose codomain is $\kw{Univ}^+$)
  and decoding using $\kw{El}^+$. 
  
  \subsection{Inductive types}
  Because \calf{} is an extensional type theory, general inductive types may be encoded by $W$-types. 
  However, in practice we will use a more natural presentation like the following: 
  \begin{align}
    \textbf{Inductive}&\; \Nat : \tpv\; \textbf{where} \label{def:nat}\\
    \kw{zero} &: \Nat \nonumber\\
    \kw{suc} &: \Nat \to \Nat \nonumber
  \end{align}
  The definition above may be elaborated into the following $W$-type:
  \begin{align*}
    \Nat &= \kw{El}^+(\widehat{\Nat})\\
    \widehat{\Nat} &: \kw{Univ}^+\\ 
    \widehat{\Nat} &= \widehat{W}(\widehat{2}, \widehat{B})\\ 
    \widehat{B} &: \tmv{2} \to \kw{Univ}^+\\ 
    B(b) &= \kw{if}(b, \widehat{0}, \widehat{1})
  \end{align*}
  In general a declaration like \cref{def:nat} should be thought of defining the \emph{code} of an inductive type. 
  Moreover, because inductive \emph{families} can be defined using \emph{indexed containers} \citep{altenkirch-ghani-hancock-mcbride-morris:2015}, 
  which in turn can be defined using $W$-types, we will also use a similar notation for defining inductive families; 
  the precise schema of inductive families and elaboration procedure is beyond the scope of our work, and we defer to the relevant literature 
  for details. 

  \subsection{Uniqueness of cost bounds}

  The theory of cost effects introduced in \citet{niu-sterling-grodin-harper:2022} is sufficient to 
  define and \emph{compose} cost bounds of programs (see \cref{sec:cost-effect}). 
  However, in order to prove adequacy, we have to 
  be able to go the other way: it needs to be the case that a given cost bound may be shown to be \emph{unique}. 
  We axiomatize uniqueness as follows: 
  \begin{align*}
    \kw{step}/\kw{inj} &: \impl{A, (a , a' : A) (c, c' : \mathbb{C})} 
      \mstep{c}{\ret{a}} = \mstep{c'}{\ret{a'}} \to a = a' \times \Cl (c = c')
  \end{align*}
  Note that because the premise of $\kw{step}/\kw{inj}$ could have been derived using a proof of the extensional phase, 
  we must seal the equation $c = c'$ by the intensional modality. 
  We will show that $\kw{step}/\kw{inj}$ holds in an extension of the counting model of \calf{} in \cref{sec:model-calfpp}.

%% file: stlc.tex
\section{Warm-up: STLC} \label{sec:STLC}

In this section, we define and study a cost-aware denotational semantics for the STLC. 
In the following, we suppress some notation from meta-level terms, \ie{} we write $e : A$ for $e : \tmv{A}$. 

\subsection{Representing object languages in \calfpp{}}

The exact mechanism by which object-level syntax is defined is immaterial for our purposes; 
we may choose from a variety of first-order encodings definable using inductive types/families. 
As an example, we will present an intrinsically-typed nameless representation for STLC based on \citet{benton-hur-kennedy-mcbride:2012}. 

\paragraph{Notation}

In this paper we write \eg{} \code{bool} for object-level syntactic phrases. 

\subsection{Syntax of the STLC}\label{sec:stlc-syn}

We consider a version of STLC with a base type of observations $\code{bool}$ with two points $\code{tt}, \code{ff} : \code{bool}$: 
\begin{align*}
  \textbf{Inductive}\; &\kw{Ty} : \tpv\; \textbf{where}\\ 
    \code{bool} &: \kw{Ty}\\
    \Rightarrow &: \kw{Ty} \to \kw{Ty}
\end{align*} 
Because we work with an intrinsic encoding, the type of terms is indexed by an object-level context 
$\kw{Con} \coloneqq \kw{list}(\kw{Ty})$ and type: 
\begin{align*} 
  \textbf{Inductive}\; &\kw{Tm} : \kw{Con} \to \kw{Ty} \to \tpv\; \textbf{where}\\ 
    \code{var} &: \impl{\Gamma, A} \kw{Var}(\Gamma, A) \to \kw{Tm}(\Gamma, A) \\ 
    \code{lam} &: \impl{\Gamma, A_1, A_2} \kw{Tm}(A_1::\Gamma, A_2) \to \kw{Tm}(\Gamma, A_1 \Rightarrow A_2)\\ 
    \code{ap} &: \impl{\Gamma, A_1, A_2} \kw{Tm}(\Gamma, A_1 \Rightarrow A_2) \to \kw{Tm}(\Gamma, A_1) \to \kw{Tm}(\Gamma, A_2)\\ 
    \code{tt} &: \impl{\Gamma} \kw{Tm}(\Gamma, \code{bool})\\
    \code{ff} &: \impl{\Gamma} \kw{Tm}(\Gamma, \code{bool})
\end{align*} 
In the above, elements of the family $\kw{Var}$ represents proofs for variable indexing: 
\begin{align*} 
  \textbf{Inductive}\; &\kw{Var} : \kw{Con} \to \kw{Ty} \to \tpv\; \textbf{where}\\ 
    \code{here} &: \impl{\Gamma, A} \kw{Var}(A::\Gamma, A)\\ 
    \code{next} &: \impl{\Gamma, A, A_1} \kw{Var}(\Gamma, A_1) \to \kw{Var}(A::\Gamma, A_1)
\end{align*}

\begin{definition}[Substitution]
  A substitution from $\Gamma$ to $\Gamma'$ is defined as 
  $\kw{Sub}(\Gamma, \Gamma') \coloneqq (A : \kw{Ty}) \to \kw{Var}(\Gamma, A) \to \kw{Tm}(\Gamma', A)$.
\end{definition}

\paragraph{Notation}
Given a substitution $\sigma : \kw{Sub}(\Gamma, \Gamma')$
 and term $e : \kw{Tm}(\Gamma, A)$, 
we write $e[\sigma] : \kw{Tm}(\Gamma', A)$ for the result of the substitution. 
Given $e : \kw{Tm}(A_1::\Gamma, A_2)$ and $e' : \kw{Tm}(\Gamma, A_1)$ we also write 
$e[e'] : \kw{Tm}(\Gamma, A_2)$ for the result of substituting the first free variable of $e$ for $e'$. 

\subsubsection{Operations on substitutions}
One can extend a substitution $\sigma : \kw{Sub}(\Gamma, \Gamma')$ by a term $e : \kw{Tm}(\Gamma', A)$:
\begin{align*}  
  \kw{cons} &: \impl{\Gamma, \Gamma', A} \kw{Tm}(\Gamma', A) \to \kw{Sub}(\Gamma, \Gamma') \to \kw{Sub}(A::\Gamma, \Gamma')\\ 
  \kw{cons}(e, \sigma, A', \kw{now}) &= e\\ 
  \kw{cons}(e, \sigma, A', \kw{next}(v)) &= \sigma(v)
\end{align*}

One can also shift substitution $\sigma : \kw{Sub}(\Gamma, \Gamma')$ to account for context extensions, written 
as $\uparrow^A \sigma  : \kw{Sub}(A::\Gamma, A::\Gamma')$. 
We will use the following property about substitution: 
\begin{proposition} \label{prop:subst}
  Given $e : \kw{Tm}(A::\Gamma, A')$, $\sigma : \kw{Sub}(\Gamma, \kw{nil})$, 
  and $e' : \kw{Tm}(\kw{nil}, A)$, we have that 
  $e[\uparrow^A \sigma][e'] = e[\kw{cons}(e', \sigma)]$. 
\end{proposition}

\subsection{Operational semantics} 

We work with a call-by-value operational semantics for STLC: 
\begin{align*} 
  \textbf{Inductive}\; &\kw{Val} : \impl{A} \kw{Pg}(A) \to \tpv\; \textbf{where}\\ 
  \code{tt}/\kw{val} &: \kw{Val}(\code{tt})\\ 
  \code{tt}/\kw{val} &: \kw{Val}(\code{ff})\\ 
  \code{lam}/\kw{val} &: \impl{e} \kw{Val}(\code{lam}(e))\\\\
  \textbf{Inductive}\; &{\mapsto} : \impl{A} \kw{Pg}(A) \to \kw{Pg}(A) \to \tpv\; \textbf{where}\\ 
    \beta &: \impl{A_1, A_2, e, e_1} \code{ap}(\code{lam}(e), e_1) \mapsto e[e_1]\\ 
    \kw{ap}/\kw{l} &: \impl{A_1, A_2, e, e', e_1} e \mapsto e' \to \code{ap}(e, e_1) \mapsto \code{ap}(e', e_1)\\ 
    \kw{ap}/\kw{l} &: \impl{A_1, A_2, e, e', e_1} \kw{Val}(e) \to e_1 \mapsto e_1' \to \code{ap}(e, e_1) \mapsto \code{ap}(e, e_1') 
\end{align*} 
In the above, we write $\kw{Pg}(A) \coloneqq \kw{Tm}([], A)$ for the type of closed STLC terms.
Evaluation may be defined as the reflexive-transitive closure of $\mapsto$: 
\begin{align*}
  \textbf{Inductive}\; &\mapsto^* :\impl{A} \kw{Pg}(A) \to \kw{Pg}(A) \to \tpv\; \textbf{where}\\
  \kw{refl} &: \impl{e} e \mapsto^* e\\  
  \kw{trans} &: \impl{e} e \mapsto e_1 \to e_1 \mapsto^* e_2 \to e \mapsto^* e_2 
\end{align*}
We then define evaluation: $e \Downarrow v \coloneqq e \mapsto^* v \times \kw{Val}(v)$. 
In a similar fashion, we may define the cost-aware evaluation relation by using a 
$\Nat$-indexed version of the reflexive-transitive closure of $\mapsto$: 
$e \Downarrow^c v \coloneqq e \mapsto^{(c)} v \times \kw{Val}(v)$. 

\subsubsection{Phase-separated cost semantics}\label{sec:phase-sep-cost-sem}

As discussed in \cref{sec:cost-aware-adequacy}, 
we cannot directly use the cost-aware evaluation relation defined above in the statement of 
the cost-aware adequacy theorem. 
Instead, we define a more refined version of the cost-aware evaluation relation that 
\emph{restricts} to the ordinary evaluation relation in the extensional phase. 
To this end, we may define a \emph{phase-separated} version of the cost-aware reflexive transitive closure: 
\begin{align*}
  \textbf{Inductive}\; & {\mapsto_{\ExtOpn}} :\impl{A} \kw{Pg}(A) \to \Cl \Nat \to \kw{Pg}(A) \to \tpv\; \textbf{where}\\
  \kw{refl} &: \impl{e} e \mapsto_{\ExtOpn}^{\eta_{\Cl} 0} e\\  
  \kw{trans} &: \impl{c, e} e \mapsto e_1 \to e_1 \mapsto_{\ExtOpn}^{c} e_2 \to e \mapsto_{\ExtOpn}^{c \mathrel{(\Cl +)} \eta_{\Cl}1} e_2 
\end{align*}
Cruically, this relation becomes equivalent to the ordinary reflexive transitive closure under the extensional phase: 
\begin{proposition}
  Given $u : \ExtOpn$, we have that $e \mapsto_{\ExtOpn}^{\eta_{\Cl} c} v$ if and only if $e \mapsto^* v$ for all $c : \Nat$.
\end{proposition}
Consequently, we may define phase-separated evaluation as 
$e \Downarrow_{\ExtOpn}^c v \coloneqq e \mapsto^{(c)} v \times \kw{Val}(v)$, which satisfies a similar 
restriction property: 
\begin{proposition}\label{prop:eval-ext-equiv}
  Given $u : \ExtOpn$, we have that $e \eval^{\eta_{\Cl} c} v$ 
  if and only if $e \Downarrow v$ for all $c : \Nat$.
\end{proposition}

\subsection{A cost-aware denotational semantics for STLC}

We now define a denotational semantics for STLC in \calfpp{} based on the standard polarized decomposition of 
call-by-value. Types are interpreted as follows: 
\begin{align*}
  \semty{-} &: \kw{Ty} \to \tpv\\ 
  \semty{\code{bool}} &= \kw{bool}\\
  \semty{A_1 \Rightarrow A_2} &= \UU{\semty{A_1} \to \F{\semty{A_2}}}
\end{align*}
The interpretation for types is extended to contexts in the obvious way: 
\begin{align*}
  \semcon{-} &: \kw{Con} \to \tpv\\ 
  \semcon{\kw{nil}} &= 1\\ 
  \semcon{A :: \Gamma} &= \semty{A} \times \semcon{\Gamma}
\end{align*}
For the interpretation of terms, we insert cost effects for elimination forms to account for steps in the operational semantics: 
\begin{align*}
  \semvar{-} &: \impl{\Gamma, A} \kw{Var}(\Gamma, A) \to \semcon{\Gamma} \to \semty{A}\\ 
  \semvar{\code{here}} &= \pi_1\\ 
  \semvar{\code{next}(v)} &= \semtm{v} \circ \pi_2 \\\\ 
  \semtm{-} &: \impl{\Gamma, A} \kw{Tm}(\Gamma, A) \to \semcon{\Gamma} \to \F{\semty{A}}\\ 
  \semtm{\code{var}(v)} &= \kw{ret} \circ \semvar{v}\\ 
  \semtm{\code{tt}} &= \lambda \_.\; \ret{\kw{tt}}\\ 
  \semtm{\code{ff}} &= \lambda \_.\; \ret{\kw{ff}}\\ 
  \semtm{\code{lam}(e)} &= \lambda \gamma : \semcon{\Gamma}.\; \lambda a : \semty{A_1}.\; \semtm{e}(a, \gamma)\\ 
  \semtm{\code{ap}(e, e_1)} &= \lambda \gamma : \semcon{\Gamma}.\; 
    \bind{\semtm{e}(\gamma)}{\lambda f.\; \bind{\semtm{e_1}(\gamma)}{\lambda a.\; \mstep{1}{f(a)}}} 
\end{align*}

\subsection{Computational adequacy}\label{sec:adequacy-STLC}

\subsubsection{Logical relation for adequacy} 

Following the classic adequacy proof of Plotkin, we prove our cost-aware adequacy theorem by means of 
a logical relations construction. 
First, we define a binary logical relation relating the values of STLC of type $A : \TySTLC{}$ with 
values in the semantic domain $\semtySTLC{A}$ by induction on $A$:  
\begin{align*}
    {\approx} &: \impl{A} \kw{Pg}(A) \to \semty{A} \to \tpv\\
    \code{b} \approx_{\code{bool}} b &= (\code{b} = \overline{b})\\ 
    \code{e} \approx_{A_1 \Rightarrow A_2} e &= 
     (\Sigma \code{e_2} : \TmSTLC(A_1, A_2).\; \code{e} = \code{lam}(\code{e_2}) \\
     &\times \kw{U}(((\code{e_1} : \kw{Pg}(A_1), e_1 : \semty{A_1}) \to \code{e_1} \approx_{A_1} e_1
       \to \code{e_2}[\code{e_1}] \approx^\Downarrow_{A_2} e(e_1))))
\end{align*} 
In the above $\overline{-}$ sends $\kw{tt}$ to $\code{tt}$ and $\kw{ff}$ to \code{ff}, and $-^\Downarrow$ lifts a relation 
on values to computations using the phase-separated evaluation relation defined in \cref{sec:phase-sep-cost-sem}:  
\begin{align*}
  -^\Downarrow &: \impl{A} (\kw{Pg}(A) \to \semty{A} \to \tpv) \to (\kw{Pg}(A) \to \F{\semty{A}} \to \tpv)\\ 
  \code{e} \mathrel{R^\Downarrow} e &= 
  \Sigma \code{v} : \kw{Pg}(A).\; \Sigma v : \semty{A}.\; (\code{e} \eval^{\eta_{\Cl} c} \code{v}) \times e = \mstep{c}{\ret{v}} \times \code{v} \mathrel{R} v
\end{align*}
The relation is readily lifted to contexts: 
\begin{align*}
  \textbf{Inductive}\; &{\approx} : \impl{\Gamma} \kw{Inst}(\Gamma) \to \semcon{\Gamma} \to \tpv\; \textbf{where}\\ 
    \kw{emp} &: \kw{nil} \approx_{\cdot} \kw{\triv}\\ 
    \kw{cons} &: \impl{\Gamma, \code{\gamma}, \gamma, A, \code{a}, a} 
    \code{a} \approx_A a \to \code{\gamma} \approx_{\Gamma} \gamma \to 
    \kw{cons}(\code{a}, \code{\gamma}) \approx_{A::\Gamma} (a, \gamma)
\end{align*}

\subsubsection{Fundamental theorem}

As usual, we may prove the fundamental theorem of the logical relation by induction on terms. 
The details of the proof can be found in \cref{appendix:STLC}.

\begin{theorem}[FTLR]\label{thm:FTLR}
  Given a STLC term $e : \TmSTLC(\Gamma, A)$, if $\code{\gamma} \approx_{\Gamma} \gamma$, then 
  $e[\code{\gamma}] \approx^\Downarrow_A \semtm{e}(\gamma)$. 
\end{theorem}

\begin{corollary}[Computational adequacy] \label{coro:cost-adequacy}
  Given a closed term $e : \kw{Pg}(\code{bool})$, if $e = \mstep{c}{\ret{b}}$ for some $c : \Nat$ and $b : \kw{bool}$, then 
  $e \Downarrow^{\eta_{\Cl} c} \overline{b}$. 
\end{corollary}

\begin{proof}
  Suppose that $e : \kw{Pg}(\code{bool})$ and $e = \mstep{c}{\ret{b}}$.  
  By \cref{thm:FTLR}, we know that $e \approx^{\Downarrow}_{\code{bool}} \semtmSTLC{e}$, which means that 
  there exists $c'$, $\code{b}$, and $b'$ such that $e \eval^{\eta_{\Cl} c'} \code{b}$,
  $e = \mstep{c'}{\ret{b'}}$, and $\code{b} \approx_{\code{bool}} b'$. 
  By $\kw{step}/\kw{inj}$, we have that $b = b'$ and $\Cl(c = c')$. 
  By the definition of the logical relation at $\code{bool}$, we have that 
  $\code{b} = \overline{b'} = \overline{b}$. The result then holds because 
  $\eta_{\Cl} c' = \eta_{\Cl} c$ since $\Cl$ is a lex monad. 
\end{proof}

\begin{corollary}[Extensional adequacy]
  Suppose that $u : \ExtOpn$. Given a closed term $e : \kw{Pg}(\code{bool})$, if $e = \ret{b}$ for some $b : \kw{bool}$, 
  then $e \Downarrow \overline{b}$. 
\end{corollary}

\begin{proof}
  Because $e = \ret{b} = \mstep{0}{\ret{b}}$, we may apply \cref{coro:cost-adequacy} to obtain  $e \eval^{\eta_{\Cl} 0} \overline{b}$. But since we have $u : \ExtOpn$, 
  we have $e \Downarrow \overline{b}$ by \cref{prop:eval-ext-equiv}.
\end{proof}

%% file: iteration.tex
\section{\calfiter{}: unbounded iteration}\label{sec:iteration}

In order to define a model of Modernized Algol (see \cref{sec:MA}), we will need to extend \calfpp{} 
with facilities for modeling while loops.  
In this section, we present \calfiter{}, a metalanguage for \emph{unbounded iteration}. 
An intuitive way to think about iteration is as a coinductive system. 
If we are given a ``one step'' computation $f : A \to B + A$ in which the left summand 
represents the terminal state and the right summand represents the nonterminal state, 
an iterative computation of $f$ can be thought of as running $f$ until the terminal state is reached. 
In terms of equations this is expressed as an unfolding rule: 
$\kw{iter}(f, a) = [\kw{ret}; \kw{iter}(f)] \circ f(a)$ (here $[f ; g] : A + B \to C$ is the 
sum of $f : A \to C$ and $g : B \to C$).  

\paragraph{Lifted computations}

Although it might be tempting to combine the computational effects of 
\emph{cost} and \emph{partiality}, doing so will complicate our model construction for \calfiter{}. 
In particular, it is not immediately clear how to assign meaning to a potential divergent \emph{type} computation 
in the usual adjunction models of CBPV we consider in this paper.
Fortunately, we may sidestep this problem by axiomatizing a more general class of \emph{lifted computations} $\Lift{A}$ that supports 
possibly divergent computations and isolate among these the previously cost-sensitive (but total) computations 
$F A \hookrightarrow L A$: 
\begin{align*}
  \kw{L} &: \tpv \to \tpc\\ 
  \kw{lift} &: \impl{A} \F{A} \to \Lift{A}\\
  \kw{lift}/\kw{inj} &: \impl{A} \kw{lift}(e) = \kw{lift}(e') \to e = e'
\end{align*}
Similar to free computations, we may sequence lifted computations: 
\begin{align*}
  \kw{bind}_{\kw{L}} &: \impl{A, B} \tmc{\Lift{A}} \to (\tmv{A} \to \tmc{\Lift{B}}) \to \tmc{\Lift{B}}
\end{align*}
Note that in contrast to free computations one may only sequence a lifted computation with another lifted computation. 
The sequencing satisfies the 
expected equational laws with respect to the unit of lifting, defined as 
$\kw{ret}_{\kw{L}} \coloneqq \kw{lift} \circ \kw{ret}$. 
Moreover, the lifting of free computations commutes with sequencing and the cost effect: 
\begin{align*}
  \kw{lift}/\kw{bind} &: \impl{A, B, e, f} \kw{lift}(\bind{e}{f}) = \bindl{\kw{lift}(e)}{\kw{lift} \circ f}\\
  \kw{lift}/\kw{step} &: \impl{A, e, c} \kw{lift}(\mstep{c}{e}) = \mstep{c}{\kw{lift}(e)}
\end{align*}

\paragraph{Notation}

We write $a \leftarrow e; f(a)$ for $\bind{e}{f}$ and $a \leftarrow_{\kw{L}} e; f(a)$ 
for $\bindl{e}{f}$.

\paragraph{Propositional truncation}

In order to state the axioms governing cost decomposition for \calfiter{}, we will 
need to work with the \emph{propositional truncation} of a type \citet{hottbook}. 
Following standard notation, we write $\norm{A}$ for the propositional truncation of a given type $A$. 
As usual we define \emph{mere existence} as the propositional truncation of a dependent sum: 
$\exists a : A.\; B(a) \coloneqq \norm{\Sigma a : A.\; B(a)}$. 
Given an assumption of the form $\exists a : A.\; B(a)$, 
we say that ``there merely exists $a : A$ such that $B(a)$''. The universal property of 
propositional truncation allows one to extract the witness $a$ and associated data $B(a)$ when one 
is proving a proposition. 

\paragraph{Higher-order recursion}

Given that we have indulged in unbounded iteration, 
it is natural to ask why not also assume arbitrary fixed-points, 
from which iteration may be derived as a special case? As we will show in \cref{sec:metatheory-iter}, 
we model \emph{lifted} computations of \calfiter{} as terms of a certain partiality monad $\bot$. 
However, $\bot$ only supports recursion for \emph{continuous} functions, and 
it is not possible to enforce this property in the kind of model we use to interpret \calfpp{} and 
\calfiter{}. A possibility for adding \emph{all} fixed-points on a computation domain is 
to use \emph{synthetic domain theory}, which we discuss in \cref{sec:conclusion}.

\subsection{Axioms for iteration}

Equipped with this intuition, we may axiomatize iteration in \calfiter{} as follows; 
note that iteration is only available for lifted computations: 
\begin{align*}
  \kw{iter} &: \impl{A, B} (\tmv{A} \to \tmc{\Lift{B + A}}) \to A \to \tmc{\Lift{B}}\\ 
  \kw{iter}/\kw{unfold} &: \impl{A, B, f, a, a'} \kw{iter}(f)(a) = \bindl{f(a)}{[\retl{b}; \kw{iter}(f)]}
\end{align*}
As we will see in \cref{sec:adequacy-MA}, we need iterative computations to satisfy a certain 
compactness property, in the sense that whenever an iterative computation $\kw{iter}(f, a)$ 
has a cost bound, there is a finite prefix $\kw{seq}(f, k, a)$ that suffices for 
obtaining that cost bound: 
\begin{align*}
  \kw{seq} &: \impl{A, B} (\tmv{A} \to \tmc{\Lift{B + A}}) \to \Nat \to \tmv{A} \to \tmc{\Lift{B + A}}\\ 
  \kw{seq}(f, 0)(a) &= \retl{\inr{a}}\\
  \kw{seq}(f, k+1)(a) &= \bindl{f(a)}{[\kw{ret}_{\kw{L}} \circ \kw{inl}; \kw{seq}(f, k)]}\\ 
  \kw{iter}/\kw{trunc} &: \impl{A, B, f, a, b, c} \kw{iter}(f, a) = \mstep{c}{\retl{b}} \to\\ 
  &\norm{\Sigma k : \Nat.\; \kw{seq}(g, k)(a) =\mstep{c}{\retl{\inl{b}}}}
\end{align*}

Similar to the uniqueness axioms we introduced in \cref{sec:STLC}, we require that 
cost bounds for lifted computations are unique: 
\begin{align*}
  \kw{step}_{\kw{L}}/\kw{inj} &: \impl{A, (a , a' : A) (c, c' : \mathbb{C})} 
    \mstep{c}{\retl{a}} = \mstep{c'}{\retl{a'}} \to (a = a') \times \Cl (c = c')
\end{align*}
we will also postulate that cost bounds on lifted computations may be decomposed: 
\begin{align*}
  \kw{bind}_{\kw{L}}^{-1} &: \impl{A, B, e, f, c, b} \bindl{e}{f} = \mstep{c}{\retl{b}} \to 
  \lVert \Sigma c_1, c_2 : \mathbb{C}.\; \Sigma a : A. \; e = \mstep{c_1}{\retl{a}} \times\\
  & f(a) = \mstep{c_2}{\retl{b}} \times \Cl(c = c_1 + c_2)\rVert\\ 
  \kw{step}_{\kw{L}}^{-1} &: \impl{A, c, c_1, a} \mstep{c_1}{e} = \mstep{c}{\retl{a}} \to 
  \lVert \Sigma c_2 : \mathbb{C}.\; e = \mstep{c_2}{\retl{a}} \times \Cl(c = c_1 + c_2) \rVert
\end{align*}
Note that the corresponding laws for free computations may be derived using the 
properties of the lifting operation $\kw{lift}$: 
\begin{proposition}
  There are terms of the following types: 
  \begin{align*}
    \kw{step}/\kw{inj} &: \impl{A, (a , a' : A) (c, c' : \mathbb{C})} 
    \mstep{c}{\ret{a}} = \mstep{c'}{\ret{a'}} \to (a = a') \times \Cl (c = c')\\
    \kw{bind}^{-1} &: \impl{A, X, e, f, c, b} \bind{e}{f} = \mstep{c}{\ret{b}} \to 
    \lVert \Sigma c_1, c_2 : \mathbb{C}.\; \Sigma a : A. \; e = \mstep{c_1}{\ret{a}} \times\\
    &f(a) = \mstep{c_2}{\ret{b}} \times \Cl(c = c_1 + c_2) \rVert\\ 
    \kw{step}^{-1} &: \impl{A, c, c_1, e, a} \mstep{c_1}{e} = \mstep{c}{\ret{a}} \to 
    \lVert \Sigma c_2 : \mathbb{C}.\; e = \mstep{c_2}{\ret{a}} \times \Cl(c = c_1 + c_2) \rVert
  \end{align*}
\end{proposition}

From these axioms, we may derive some useful reasoning principles regarding cost bounds:
\begin{proposition}\label{prop:ret-zero-lift}
  There is a term of type $\mstep{c}{\retl{a}} = \retl{a'} \to \Cl(c = 0)$. 
\end{proposition}

\begin{proof}
  By $\kw{step}_{\kw{L}}^{-1}$, we have that there merely exists $c'$ such that $\retl{a} = \mstep{c'}{ \retl{a'}}$ and $\Cl(0 = c + c')$. 
  Because $\Cl(c = 0)$ is a proposition, we may project the underlying witness and data. 
  Lifting using the functorial action of $\Cl$, it suffices to show that there is a term 
  $0 = c + c' \to c = 0$, which clearly holds for a cancellative monoid.
\end{proof}

\begin{proposition}\label{prop:step-ret-lift}
  There is a term of type $\mstep{c}{\retl{a}} = \retl{a'} \to (a = a')$. 
\end{proposition}

\begin{proof}
  By \cref{prop:ret-zero-lift}, we have that $\Cl(c = 0)$. By induction principle of the closed modality, we have to consider 
  two cases. First, if $u : \ExtOpn$, then we have $\mstep{c}{\retl{a}} = \retl{a} = \retl{a'}$, from which the result 
  follows from $\kw{ret}/\kw{inj}$ and $\kw{in}/\kw{inj}$. Otherwise, we have $c = 0$, and the result holds from the same argument.
\end{proof}

%% file: algol.tex
\section{Modernized Algol}\label{sec:MA}

We define and study a denotational semantics for a variant of Modernized Algol (\MA{}) 
as formulated in \citet{harper:2012:pfpl}. \MA{} is a procedural programming language with 
first-order store and unbounded iteration and obeys a stack discipline in the sense that 
store assignables are deallocated when they go out of scope. 
A characteristic feature of \MA{} is the distinction between expressions and commmands 
that reflects the separation of mathematical and effectful computation. 
Unlike \opcit{}, we restrict our expression language to a total language 
(a mild extension of \STLC{} as presented in \cref{sec:STLC}) 
and extend the command language with a primitive iteration command. 

The version of \MA{} that we present is not as bare-bones as the \STLC{} from \cref{sec:STLC}, but it is 
still somewhat austere in terms of type structures, with natural numbers being the only 
(non trivial) inductive data type. Because our intention is to illustrate  
the construction and adequacy proof of a cost-aware denotational semantics, we have delibrately kept the 
language in question simple to isolate the core ideas. 
The methods we develop here may be extended to include richer data structures so that one may program the 
algorithms studied in \citet{niu-sterling-grodin-harper:2022} 
(in fact we can already define Euclid's algorithm in this modest version of \MA{}). 

\subsection{Syntax of \MA{}}

\MA{} is equipped with a small class of inductive types, functions types, and 
a type $\code{cmd}(A)$ of reified commands that compute expressions of type $A$ 
(see \cref{fig:types-MA}). We also outline a class of \emph{strictly positive types} 
that may be used as assignables, which restricts \MA{} to \emph{first-order stores}.
We define the type of strictly positive types as $\kw{Pos} \coloneqq \Sigma A : \TyMA{}.\; \kw{pos}(A)$. 

\paragraph{Convention} 
In order to faciliate readbility, we will not write the proof of strict positivity,
\ie{} we write $\code{bool} : \kw{Pos}$ for $(\code{bool}, \code{bool}/\kw{pos}) : \kw{Pos}$. 

\begin{figure}
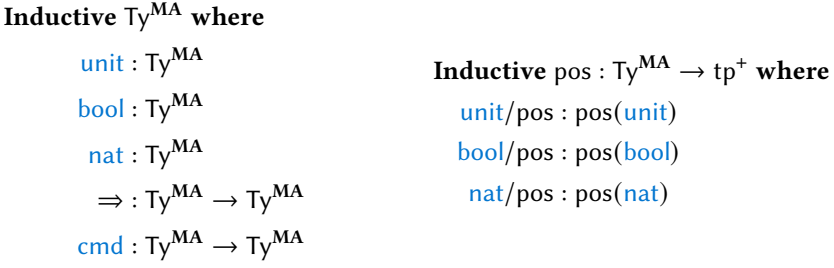

  \begin{minipage}{0.45\textwidth}
    \begin{align*}
      \textbf{Inductive}\; &\TyMA{} \; \textbf{where}\\ 
        \code{unit} &: \TyMA{}\\
        \code{bool} &: \TyMA{}\\ 
        \code{nat} &: \TyMA{}\\
        \Rightarrow &: \TyMA{} \to \TyMA{}\\
        \code{cmd} &: \TyMA{} \to \TyMA{}
    \end{align*} 
  \end{minipage}
  \begin{minipage}{0.45\textwidth}
    \begin{align*}  
      \textbf{Inductive}\; &\kw{pos} : \TyMA{} \to \tpv \; \textbf{where}\\ 
      \code{unit}/\kw{pos} &: \kw{pos}(\code{unit})\\ 
      \code{bool}/\kw{pos} &: \kw{pos}(\code{bool})\\ 
      \code{nat}/\kw{pos} &: \kw{pos}(\code{nat})
    \end{align*}
  \end{minipage}
  \caption{Left: types of \MA{}; right: strictly positive types.}
  \label{fig:types-MA}
\end{figure}

As for the \STLC{}, we work with an instrinsic encoding of \MA{}; the terms of \MA{} 
are presented in \cref{fig:ma-statics}. 
Both the judgment for well-typed commands $\Cmd$ and expressions $\TmMA$ are indexed by 
a \emph{signature} $\kw{Sig} \coloneqq \kw{list}(\kw{Pos})$ of strictly positive types. 
Well-typed expressions $e : \TmMA(\Sigma, \Gamma, A)$ are written as $\Gamma \vdash_\Sigma e : A$ and 
well-typed commands $m : \Cmd(\Sigma, \Gamma, A)$ are written as $\Gamma \vdash_\Sigma m \div A$;   
note that we make use of a mutual inductive definition because well-typed expressions and 
commands must be defined simultaneously. 

\begin{figure}
  \begin{align*}  
  \textbf{Inductive}\; &\kw{Var} : \kw{Ctx} \to \TyMA \to \tpv\; \textbf{where}\\ 
  \textbf{Inductive}\;&\Cmd : \kw{Sig} \to \kw{Ctx} \to \TyMA{} \to \tpv \; \textbf{where}\\
  \textbf{mutual}\;&\TmMA : \kw{Sig} \to \kw{Ctx} \to \TyMA{} \to \tpv \; \textbf{where}
  \end{align*}
  \begin{mathpar}
    \inferrule{
      v : \kw{Var}(\Gamma, A)
    }{
      \Gamma \vdash_\Sigma \code{var}(v) : A 
    }

    \inferrule{
    }{
      \Gamma \vdash_\Sigma \code{\triv} : \code{unit}
    }

    \inferrule{
    }{
      \Gamma \vdash_\Sigma \code{zero} : \code{nat}
    }

    \inferrule{
      \Gamma \vdash_\Sigma e : \code{nat}
    }{
      \Gamma \vdash_\Sigma \code{suc}(e) : \code{nat} 
    }

    \inferrule{
      \Gamma \vdash_\Sigma e : \code{nat}\\
      \Gamma \vdash_\Sigma e_1 : A\\
      \code{nat}::\Gamma \vdash_\Sigma e_2 : A
    }{
      \Gamma \vdash_\Sigma \code{ifz}(e, e_1, e_2) : A
    }

    \inferrule{
    }{
      \Gamma \vdash_\Sigma \code{tt}, \code{ff} : \code{bool}
    }
  
    \inferrule{
      A_1::\Gamma \vdash_\Sigma e : A_2
    }{
      \Gamma \vdash_\Sigma \code{lam}(e) : A_1 \Rightarrow A_2
    }
  
    \inferrule{
      \Gamma \vdash_\Sigma e : A_1 \Rightarrow A_2\\
      \Gamma \vdash_\Sigma e_1 : A_1
    }{
      \Gamma \vdash_\Sigma \code{ap}(e, e_1) : A_2
    }
  
    \inferrule{
      \Gamma \vdash_{\Sigma} m \div A
    }{
      \Gamma \vdash_{\Sigma} \code{cmd}(m) : \code{cmd}(A)
    }
  
    \inferrule{
      \Gamma \vdash_\Sigma a : A
    }{
      \Gamma \vdash_\Sigma \code{ret}(a) : \code{cmd}(A)
    }
  
    \inferrule{
      \Gamma \vdash_\Sigma e : \code{cmd}(A)\\ 
      A::\Gamma \vdash_\Sigma m \div B
    }{
      \Gamma \vdash_\Sigma \code{bnd}(e, m) \div B
    }
  
    \inferrule{
      \Sigma[n] = (\code{bool} , -)\\
      \Gamma \vdash_\Sigma m \div \code{unit}
    }{
      \Gamma \vdash_\Sigma \code{while}[n](m) \div \code{unit}
    }
    
    \inferrule{
      \Sigma[n] = (A , -)
    }{
      \Gamma \vdash_{\Sigma} \code{get}[n] \div A
    }
  
    \inferrule{
      \Sigma[n] = (A , -)\\
      \Gamma \vdash_\Sigma e : A
    }{
      \Gamma \vdash_{\Sigma} \code{set}[n](e) \div \code{bool}
    }
  
    \inferrule{
      A_{\kw{pos}} : \kw{pos}(A)\\ 
      \Gamma \vdash_{\Sigma} e : A\\ 
      \Gamma \vdash_{(A, A_{\kw{pos}})::\Sigma} m \div \code{bool}
    }{
      \Gamma \vdash_\Sigma \code{dcl}(e, m) : \code{bool}
    }
  \end{mathpar}
   \caption{Expressions and commands of \MA{}.} 
   \label{fig:ma-statics}
  \end{figure}

\subsection{Preorder of signatures}\label{sec:preorder}

In the possible worlds semantics/Kripke models of mutable store, 
the interpretation of signature-indexed commands and expressions 
can be thought of as a family of models linked by a 
contravariant action on the \emph{preorder relation} on the signatures, 
which represents the stability of the interpretation with respect to allocation of 
new assignables. We may define the preorder on signatures as follows: 
\begin{align*}
  \textbf{Inductive}\;& \ge : \kw{Con} \to \kw{Con} \to \tpv\; \textbf{where}\\  
    \kw{refl} &: \impl{\Sigma} \Sigma \ge \Sigma\\ 
    \kw{mono} &: \impl{\Sigma, \Sigma', A} \Sigma' \ge \Sigma \to A::\Sigma' \ge A::\Sigma\\
    \kw{extend} &: \impl{\Sigma, \Sigma', A} \Sigma' \ge \Sigma \to A::\Sigma' \ge \Sigma
\end{align*}

In other words, we have a proof of $\Sigma' \ge \Sigma$ whenever $\Sigma$ occurs as a 
\emph{subsequence} of $\Sigma'$. 

\begin{proposition}
  The relation $\ge$ is reflexive and transitive. We write $\kw{tr} : \impl{\Sigma'', \Sigma', \Sigma : \kw{Sig}} \Sigma'' \ge \Sigma' \to \Sigma' \ge \Sigma \to \Sigma'' \ge \Sigma$ 
  for the proof of transitivity. 
\end{proposition}

\paragraph{Action of the preorder of signatures}

First, we show that one may shift assignables along a signature extension: 
\begin{proposition}
  There is a map $\kw{sh} : \impl{\Sigma, \Sigma'} \Sigma' \ge \Sigma \to \Nat \to \Nat$  
\end{proposition}

\begin{proposition}
  There is a map $\Uparrow : \impl{\Sigma, \Sigma', \Gamma, A} \Sigma' \ge \Sigma \to \TmMA(\Sigma, \Gamma, A) \to \TmMA(\Sigma', \Gamma, A)$.  
\end{proposition}

\begin{proposition} 
  There is a map $\Uparrow : \impl{\Sigma, \Sigma', \Gamma, A} \Sigma' \ge \Sigma \to \Cmd(\Sigma, \Gamma, A)\to \Cmd(\Sigma', \Gamma, A)$.  
\end{proposition}

\subsubsection{Substitution}

\begin{definition}[Substitution]\label{def:subst-MA}
  Given a signature $\Sigma$, a substitution from $\Gamma$ to $\Gamma'$ is defined as 
  $\kw{Sub}_\Sigma(\Gamma, \Gamma') \coloneqq (A : \kw{Ty}) \to \kw{Var}(\Gamma, A) \to \kw{Tm}(\Sigma, \Gamma', A)$.
\end{definition}

The action of the preorder of signatures extends to substitions: 
\begin{proposition}
  There is a map ${\Uparrow} : \impl{\Sigma, \Sigma', \Gamma, \Gamma'} \Sigma' \ge \Sigma \to \kw{Sub}_\Sigma(\Gamma, \Gamma') \to \kw{Sub}_{\Sigma'}(\Gamma, \Gamma')$. 
\end{proposition}

We write the following for the application of a subsititution to an expression and a command:
\begin{align*}
-[-] &: \impl{\Sigma, \Gamma, A} \TmMA(\Sigma, \Gamma, A) \to \kw{Sub}_\Sigma(\Gamma, \Gamma') \to  \TmMA(\Sigma, \Gamma', A)\\
-[-] &: \impl{\Sigma, \Gamma, A} \Cmd(\Sigma, \Gamma, A) \to \kw{Sub}_\Sigma(\Gamma, \Gamma') \to  \Cmd(\Sigma, \Gamma', A)
\end{align*}

\subsection{Operational semantics of \MA{}}

The operational semantics of \MA{} is defined separately for expressions and commands. 
Expressions execute via substitution as in \STLC{}: 
\begin{align*}
  \kw{val} &: \impl{\Sigma, A}  \kw{Pg}(\Sigma, A) \to \tpv\\ 
  {\mapsto} &: \impl{\Sigma, A} \kw{Pg}(\Sigma, A) \to \kw{Pg}(\Sigma, A) \to \tpv 
\end{align*}
In the above, we defined closed expressions as $\kw{Pg}(A) \coloneqq \TmMA(\Sigma, \kw{nil}, A)$. 
Define $\kw{Val}(\Sigma, A) \coloneqq \Sigma e : \kw{Pg}(\Sigma, A).\; \kw{val}(e)$. 
In contrast to expressions, commands execute in conjunction with a 
\emph{store} that contains values associated to the declared assignables. More explicitly, we may 
define a store as the following family indexed in a signature: 
\begin{align*} 
  \kw{store} &: \kw{Sig} \to \tpv\\ 
  \kw{emp} &: \kw{Store}(\kw{nil})\\ 
  \kw{extend} &: \impl{\Sigma, A} (p:\kw{pos}(A)) \to \kw{Val}(\kw{nil}, A) \to \kw{Store}(\Sigma) \to \kw{Store}((A, p)::\Sigma)
\end{align*}
A \emph{state} is composed of a store and command, defined as 
$\kw{State}(\Sigma, A) \coloneqq \kw{Store}(\Sigma) \times \kw{Cmd}(\Sigma, A)$, and 
we have the following judgments governing the dynamics of states: 
\begin{align*} 
  \kw{final} &: \impl{\Sigma, A}  \kw{State}(\Sigma, A) \to \tpv\\ 
  {\Mapsto} &: \impl{\Sigma, A} \kw{State}(\Sigma, A) \to \kw{State}(\Sigma, A) \to \tpv 
\end{align*}

\begin{figure}
  \begin{mathpar}
    \inferrule{
      \mu[n] = \code{ff}
    }{
      (\mu, \code{while}[n](m)) \Mapsto (\mu, \code{ret}{\code{\triv}})
    }

    \inferrule{
      \mu[n] = \code{tt}
    }{
      (\mu, \code{while}[n](m)) \Mapsto (\mu, \code{bnd}(\code{cmd}(m), \kw{wk}(\code{while}[n](m))))
    }
  \end{mathpar}
  \label{fig:dynamics-MA}
  \caption{Selected rules for the one-step transition of commands of \MA{}; 
  here $\kw{wk}$ denotes weakening of a term.}
\end{figure}
As an example, we give the rules for \code{while} in \cref{fig:dynamics-MA}.  
For brevity we have suppressed the definitions of the other commands; 
the complete definition may be found in \citet{harper:2012:pfpl}. 
As usual we define evaluation of expressions as 
$e \Downarrow_{\kw{exp}} v \coloneqq e \mapsto^* v \times \kw{val}(v)$ and 
evaluation of commands as 
$(\mu, m) \Downarrow_{\kw{cmd}} (\mu', m') \coloneqq (\mu, m) \Mapsto^* (\mu', m') \times \kw{final}(m')$. 

\subsubsection{Phase-separated operational semantics}

Similar to the case for \STLC{}, in order to state and prove adequacy, 
we will need a phase-separated version of the transition relation for both expressions and 
commands where the evaluation cost is \emph{sealed} by the closed modality: 
\begin{align*}
  {\mapsto_{\ExtOpn}} &:\impl{\Sigma, A} \kw{Pg}(\Sigma, A) \to \Cl \Nat \to \kw{Pg}(\Sigma, A) \to \tpv\\
  {\Mapsto_{\ExtOpn}} &:\impl{\Sigma, A} \kw{State}(\Sigma, A) \to \Cl \Nat \to \kw{State}(\Sigma, A) \to \tpv
\end{align*}
Using these phase-separated relations we may define phase-separated evaluation for both 
expressions and commands as 
$e \evalexp^{c} v \coloneqq e \mapsto_{\ExtOpn{}}^c v \times \kw{val}(v)$
and 
$(\mu, m) \evalcmd^c (\mu', m') \coloneqq (\mu, m) \Mapsto_{\ExtOpn}^c (\mu', m') \times \kw{final}(m')$
respectively. As before, phase-separated evaluation restricts to ordinary evaluation in the 
extensional phase: 
\begin{proposition}
  Given $u : \ExtOpn$ and $c : \Nat$, we have that $e \evalexp^{\eta_{\Cl} c} v$ is 
  equivalent to $e \Downarrow v$ and that $(\mu, m) \evalcmd^{\eta_{\Cl} c} (\mu', m')$ is 
  equivalent to $(\mu, m) \Downarrow_{\kw{cmd}} (\mu', m')$. 
\end{proposition}

\subsection{A denotational model for \MA{}}\label{sec:model-MA}

As mentioned in \cref{sec:preorder}, our denotational semantics of \MA{} is based on 
a possible-worlds model of allocation. 
Consequently, we interepret (closed) commands as 
\emph{families} of functions that may be executed on any future semantic store
(according the the preorder of signatures). 
In a language with higher-order store this definition becomes circular because  
the interpretation of signatures depends on \emph{all} types and so in particular command types; 
consequently, more refined techniques such as (abstract) step-indexing 
\citep{bmss:2011} is required to resolve the circularity of the definition. 
Because \MA{} only admits first-order store, we may  
bootstrap the definition by first defining the meaning of strictly positive types, 
which is independent of the interpretation of signatures (see \cref{fig:types-pos}). 
The definition of types is then allowed to reference the meaning of signatures: 
\begin{align*}
  \semtyMA{-} &: \TyMA \to \kw{Sig} \to \tpv\\ 
  \semtyMA{\code{unit}}(\Sigma) &= 1\\
  \semtyMA{\code{bool}}(\Sigma) &= \kw{bool}\\
  \semtyMA{\code{nat}}(\Sigma) &= \Nat\\
  \semtyMA{A_1 \Rightarrow A_2}(\Sigma) &= \UU{(\Sigma' : \kw{Con}) \to \Sigma' \ge \Sigma \to \semtyMA{A_1}(\Sigma') \to \F{\semtyMA{A_2}(\Sigma')}}\\ 
  \semtyMA{\kw{cmd}(A)}(\Sigma) &= \UU{(\Sigma' : \kw{Con}) \to \Sigma' \ge \Sigma \to \semsigMA{\Sigma'} \to \Lift{\semtyMA{A}(\Sigma') \times \semsigMA{\Sigma'}}}
\end{align*}

\begin{figure}
  \begin{minipage}{0.3\textwidth}
    \begin{align*} 
      \semtyPosMA{-} &: \kw{Pos} \to \tpv\\ 
      \semtyPosMA{(\code{unit}, \code{unit}/\kw{base})} &= 1\\
      \semtyPosMA{(\code{bool}, \code{bool}/\kw{base})} &= \kw{bool}\\
      \semtyPosMA{(\code{nat}, \code{nat}/\kw{base})} &= \Nat
    \end{align*} 
  \end{minipage}
  \begin{minipage}{0.3\textwidth}
    \begin{align*}
      \semsigMA{-} &: \kw{Sig} \to \tpv\\ 
      \semsigMA{\cdot} &= 1\\ 
      \semsigMA{A :: \Sigma} &= \semtyPosMA{A} \times \semsigMA{\Sigma} 
    \end{align*}
  \end{minipage}
  \begin{minipage}{0.3\textwidth}
    \begin{align*}
      \semconMA{-} &: \kw{Con} \to \kw{Sig} \to \tpv\\  
      \semconMA{\cdot}(\Sigma) &= 1\\ 
      \semconMA{A::\Gamma}(\Sigma) &= \semtyMA{A}(\Sigma) \times \semconMA{\Sigma}
    \end{align*}
  \end{minipage}
  \label{fig:types-pos}
  \caption{Top left: the interpretation of strictly positive types; top right: interpretation of signatures;
  bottom: interpretation of contexts.}
\end{figure}

Similar to the action of preorders at the syntactic level, 
we may define an analogous action on the interpretation of types: 
\begin{proposition} 
  There is a map 
  $\kw{up} : \impl{A, \Sigma', \Sigma} \Sigma' \ge \Sigma \to \semtyMA{A}(\Sigma) \to \semtyMA{A}(\Sigma')$. 
\end{proposition}

\begin{corollary}
  There is a map 
  $\kw{up} : \impl{\Gamma, \Sigma', \Sigma} \Sigma' \ge \Sigma \to \semconMA{\Gamma}(\Sigma) \to \semconMA{\Gamma}(\Sigma')$. 
\end{corollary}

\subsubsection{The model}

In \cref{fig:model-exp,fig:model-cmd} we present the denotational semantics of 
expression and commands of \MA{}. 
The expression-level denotational semantics for \MA{} is defined in a similar style to \cref{sec:STLC}; 
reified commands are defined using the mutually recursive interpretation of commands. 
Commands are defined using a possible-worlds semantics of first-order store. 
Observe that both the interpretation of expressions and commands at a world (\ie{} signature) 
$\Sigma$ are paramterized by a \emph{future world} $\Sigma' \ge \Sigma$. 
Consequently commands of a type $A$ 
are \emph{families} of \emph{lifted} transformations 
$\semtyMA{A}(\Sigma') \to \Lift{\semtyMA{A}(\Sigma') \times \semsigMA{\Sigma'}}$ 
of semantic stores that also produces a value of type $A$ (at the extended signature). 
In contrast, because expressions are total and cannot modify the store, they are simply interpreted 
as \emph{free} computations of the given type. 
As in the case for the \STLC{}, we must insert \kw{step}'s in both the interpretation of 
expressions and commands when $\beta$-reductions occur in the operational semantics. 

Here, we highlight the the fact that the while loop of \MA{} is interpreted as a ``while loop'' in 
\calfiter{}, a feature typical of synthetic denotational semantics. As we shall see in \cref{sec:adequacy-MA},
this synthetic interpretation allows for a somewhat involved but elementary proof of 
computational adequacy. 

\begin{small}
\begin{figure}
  \begin{align*}
    \semvar{-} &: \impl{\Sigma, \Gamma, A} \kw{Var}(\Gamma, A) \to 
    \to (\Sigma' : \kw{Sig}) \to \Sigma' \ge \Sigma \to \semconMA{\Gamma}(\Sigma') \to \semty{A}(\Sigma')\\ 
    \semvar{\code{here}}(\Sigma', p, \gamma') &= \pi_1(\gamma')\\ 
    \semvar{\code{next}(v)}(\Sigma', p, \gamma') &= (\semtm{v}(\Sigma', p) \circ \pi_2) \gamma' \\\\ 
    \semexpMA{-} &: \impl{\Sigma, \Gamma, A} \TmMA(\Sigma, \Gamma, A) \to (\Sigma' : \kw{Sig}) \to 
    \Sigma' \ge \Sigma \to \semconMA{\Gamma}(\Sigma') \to \F{\semtyMA{A}(\Sigma')}\\ 
    \semexpMA{\code{var}(v)} &= \kw{ret} \circ \semvar{v}\\ 
    \semexpMA{\code{tt}}(\Sigma', p, -) &= \ret{\kw{tt}}\\ 
    \semexpMA{\code{ff}}(\Sigma', p, -) &= \ret{\kw{ff}}\\ 
    \semexpMA{\code{zero}}(\Sigma', p, -) &= \ret{0}\\ 
    \semexpMA{\code{suc}}(\Sigma', p, \gamma') &= \bind{\semexpMA{e}(\Sigma', p, \gamma')}{\lambda n : \Nat.\; \ret{n+1}}\\ 
    \semexpMA{\code{ifz}(e, e_1, e_2)}(\Sigma', p, \gamma') &=
      n \leftarrow \semexpMA{e}(\Sigma', p, \gamma');\\ 
      &\kw{if}(n, \mstepAlert{1}{\semexpMA{e_1}(\Sigma', p, \gamma')}, \lambda n'.\; \mstepAlert{1}{\semexpMA{e_2}(\Sigma', p, (n', \gamma'))})\\
    \semexpMA{\code{lam}(e)}(\Sigma', p, \gamma') &= \ret{\lambda \Sigma'', p', \lambda a : \semtyMA{A_1}(\Sigma'').\; 
      \semexpMA{e}(\Sigma'', \kw{tr}(p', p), (a, \kw{up}(p', \gamma')))}\\ 
    \semexpMA{\code{ap}(e, e_1)}(\Sigma', p, \gamma') &= 
      f \leftarrow \semexpMA{e}(\Sigma', p, \gamma'); 
      a \leftarrow \semexpMA{e_1}(\Sigma', p, \gamma);
      \mstepAlert{1}{f(\Sigma', \kw{refl}, a)} \\
    \semexpMA{\code{cmd}(m)}(\Sigma', p, \gamma') &= 
    \ret{\lambda \Sigma'', p', (\sigma'' : \semsigMA{\Sigma''}).\; \semcmdMA{m}(\Sigma'', \kw{tr}(p', p), \kw{up}(p', \gamma'), 
    \sigma'')}
  \end{align*}  
  \caption{The interpretation of expressions of \MA{}. }
  \label{fig:model-exp}
\end{figure}
\end{small}

\begin{small}
\begin{figure}
  \begin{align*}
    \semcmdMA{-} &: \impl{\Sigma, \Gamma, A} \TmMA(\Sigma, \Gamma, A) \to (\Sigma' : \kw{Sig}) \to \Sigma' \ge \Sigma \to\\
     &\semconMA{\Gamma}(\Sigma') \to \semsigMA{\Sigma'} \to
     \Lift{\semtyMA{A}(\Sigma') \times \semsigMA{\Sigma'}}\\  
    \semcmdMA{\code{ret}(a)}(\Sigma', p, \gamma', \sigma') &= 
    \bind{\semexpMA{a}(\Sigma', p, \gamma')}{\lambda a : \semty{A}(\Sigma').\; \retl{a, \sigma'}}\\ 
    \semcmdMA{\code{bnd}(e, m)}(\Sigma', p, \gamma', \sigma') &= 
      m_1 \leftarrow_{\kw{L}} \kw{lift}(\semexpMA{e}(\Sigma', p, \gamma'));
      (a, \sigma_1) \leftarrow_{\kw{L}} m_1(\Sigma', \kw{refl}, \sigma');\\
      &\mstepAlert{1}{\semcmdMA{m}(\Sigma', p, (a, \gamma'), \sigma_1)}\\ 
    \semcmdMA{\code{while}[n](m)}(\Sigma', p, \gamma', \sigma') &= \kw{iter}(g)(\sigma') \; \textbf{where}\\ 
    g &: \semsigMA{\Sigma'} \to \Lift{(1 \times \semsigMA{\Sigma'}) + \semsigMA{\Sigma'}}\\ 
    g(\sigma) \textbf{ with } \sigma[n] &\\
    \dots \mid \kw{ff} &= \mstepAlert{1}{\retl{\inl{\triv, \sigma}}}\\ 
    \dots \mid \kw{tt} &= (-, \sigma') \leftarrow_{\kw{L}} \semcmdMA{m}(\Sigma', p, \gamma', \sigma); \mstepAlert{2}{\ret{\inr{\sigma'}}}\\ 
    \semcmdMA{\code{dcl}(e, m)} \impl{A = A} (\Sigma', p, \gamma', \sigma') &= 
      a \leftarrow_{\kw{L}} \kw{lift}(\semexpMA{e}(\Sigma', p, \gamma'));\\
      (b, (-, \sigma_1)) &\leftarrow_{\kw{L}} \semcmdMA{m}(A::\Sigma', \kw{mono}(p), \kw{up}(\kw{extend}(p), \gamma'), (a, \sigma')); 
      \mstepAlert{1}{\retl{b, \sigma_1}}\\
    \semcmdMA{\code{get}[n]}(\Sigma', p, \gamma', \sigma') &= \mstepAlert{1}{\retl{\sigma'[n], \sigma'}}\\ 
    \semcmdMA{\code{set}[n](e)}(\gamma, \sigma') &= 
      a \leftarrow_{\kw{lift}} \kw{lift}(\semexpMA{e}(\Sigma', p, \gamma', \sigma')); 
      \mstepAlert{1}{\retl{\sigma'[n], \sigma'[n \mapsto a]}}
  \end{align*}  
  \caption{The interpretation of commands of \MA{}.}
  \label{fig:model-cmd}
\end{figure}
\end{small}

\subsection{Computational adequacy}\label{sec:adequacy-MA}

In this section we prove that the denotational semantics of \MA{} defined in 
\cref{sec:model-MA} satisfies the cost-aware computational adequacy theorem from 
\cref{sec:cost-aware-adequacy} with respect to its phase-separated operational semantics. 
As in the case of the \STLC{}, we employ the method of logical relations to prove this result. 
As foreshadowed by the construction of the model in \cref{sec:model-MA}, 
we have to stage the definition of the logical relation by first defining 
the relation for strictly positive types in \cref{sec:logical-pos}, 
using this relation to define the \emph{pre}logical relation for commands in \cref{sec:prelogical-commands},
using this to define the logical relation for expressions in \cref{sec:logical-exp}, and 
finally closing the loop by defining the logical relation for commands in \cref{sec:logical-cmd}. 
We prove the fundamental theorem of the logical relation in \cref{sec:ftlr-MA}. 

\subsubsection{Logical relation for strictly positive types}\label{sec:logical-pos}

For strictly positive types, we relate the values in the semantic domain with 
their numerals in \MA{}: 
\begin{align*} 
  \textbf{Inductive}\; &{\approx} : \impl{A : \kw{Pos}} \kw{Pg}(\cdot, A) \to \semtyPosMA{A} \to \tpv \;\textbf{where}\\ 
    \code{unit}/\kw{base} &: \code{\triv} \approx_{\code{unit}} \triv\\
    \code{bool}/\kw{base} &: (b : \kw{bool}) \to \overline{b} \approx_{\code{bool}} b\\
    \code{nat}/\kw{base} &: (n : \Nat) \to \overline{n} \approx_{\code{nat}} n
\end{align*}

\subsubsection{Prelogical relation for commands}\label{sec:prelogical-commands}

Using the logical relation for strictly positive types, we may 
define the logical relation between syntatic stores and semantic stores: 
\begin{align*}
  \textbf{Inductive}\; {\sim} &: \impl{\Sigma} \kw{Store}(\Sigma) \to \semsigMA{\Sigma} \to \tpv\; \textbf{where}\\ 
    \kw{base} &: \kw{emp} \sim_{\kw{nil}} \triv\\ 
    \kw{extend} &: \impl{\Sigma, A, \code{a}, a, \mu, \sigma} \to (h : \kw{val}(\code{a})) \to 
      \code{a} \approx_{A} a \to \mu \sim_{\Sigma} \sigma \to 
       ((\code{a}, h) :: \mu) \sim_{A::\Sigma} (a, \sigma)
\end{align*}
The prelogical relation for commands is defined relative to a given relation for expressions: 
\begin{align*}  
  \kw{cmd} &: \impl{\Sigma, A}  
  (\kw{Pg}(\Sigma, A) \to \semtyMA{A}(\Sigma) \to \tpv) \to\\
  &\kw{Cmd}(\Sigma, A) \to (\semsigMA{\Sigma} \to \F{\semtyMA{A}(\Sigma) \times \semsigMA{\Sigma}}) \to \tpv\\ 
  \code{m} \mathrel{\kw{cmd}_{\Sigma, A}(R)} m &= 
  \kw{U}(\Pi \sigma, \sigma' : \semsigMA{\Sigma}, c : \Nat, a : \semtyMA{A}(\Sigma).\\ 
  &\quad m(\sigma) = \mstep{c}{\retl{a, \sigma'}} \to\\
  &\quad \Pi \mu : \kw{Store}(\Sigma) \to \mu \sim_{\Sigma} \sigma \to\\
  &\quad \Sigma \code{a} : \kw{Pg}(\Sigma, A), \mu' : \kw{Store}(\Sigma).
  (\mu, \code{m}) \evalcmd^{\eta_{\Cl} c} (\mu', \code{ret}(\code{a})) \times
  \code{a} \mathrel{R} a \times \mu' \sim_{\Sigma} \sigma')
\end{align*}

Roughly, given a relation $R$ between syntatic and semantic values, 
a syntactic command $\code{m}$ is related to a semantic command $m$ 
when given logically related syntactic and semantic stores $\mu$ and $\sigma$, 
we have that if semantically $m(\sigma)$ computes to a semantic value $v$ and store $\sigma'$ incurring some cost,
then executing $(\mu, \code{m})$ will evaluate with the same cost 
to a syntactic value $\code{v}$ such that $R(\code{v}, v)$ 
and a syntactic store $\mu'$ related to the semantic store $\sigma'$. 

\subsubsection{Logical relation for expressions}\label{sec:logical-exp}

The logical relation for expressions may be defined as in the case of \STLC{}, using the 
prelogical relation of commands defined in \cref{sec:prelogical-commands} in the case of reified commands: 
\begin{align*}
  {\approx} &: \impl{\Sigma, A} \kw{Pg}(\Sigma, A) \to \semtyMA{A}(\Sigma) \to \tpv\\ 
    \code{u} \approx_{\Sigma, \code{unit}} u &= (\code{u} = \code{\triv})\\ 
    \code{b} \approx_{\Sigma, \code{bool}} b &= (\code{b} = \overline{b})\\ 
    \code{n} \approx_{\Sigma, \code{nat}} n &= (\code{n} = \overline{n})\\ 
    \code{e} \approx_{\Sigma, A_1 \Rightarrow A_2} e &= 
    \Sigma \code{e_2} : \TmMA(\Sigma, A_1, A_2).\; \code{e} = \code{lam}(\code{e_2})\\ 
    &\quad \times \kw{U}(\Pi \Sigma'.\; \Pi p : \Sigma' \ge \Sigma. \; 
      \Pi (\code{e_1} : \kw{Pg}(\Sigma', A_1), e_1 : \semtyMA{A_1}(\Sigma')).\;\\ 
      &\quad\quad\code{e_1} \approx_{\Sigma', A_1} e_1 \to
       (\code{\uparrow} (p, \code{e}))[\code{e_1}] \approx^\Downarrow_{\Sigma', A_2} e(\Sigma', p, e_1))\\
    \code{e} \approx_{\Sigma, \code{cmd}(A)} m &= \Sigma \code{m} : \Cmd(\Sigma, A).\; 
    \code{e} = \code{cmd}(\code{m})\\
    &\times \kw{U}(\Pi \Sigma'.\; \Pi p : \Sigma' \ge \Sigma. \;  
    (\Uparrow^p \code{m}) \mathrel{\kw{cmd}_{\Sigma', A}(\approx_{\Sigma', A})} m(\Sigma', p))
\end{align*} 
Following \cref{sec:STLC}, the operator $-^\Downarrow$ lifts a relation on values to computations: 
\begin{align*}
  -^\Downarrow &: \impl{\Sigma, A} (\kw{Pg}(\Sigma, A) \to \semtyMA{A}(\Sigma) \to \tpv) \to 
                  (\kw{Pg}(\Sigma, A) \to \F{\semtyMA{A}(\Sigma)} \to \tpv)\\ 
  \code{e} \mathrel{R^\Downarrow} e &= 
  \UU{\Pi\; v : \semtyMA{A}(\Sigma), c : \Nat.\; 
    e = \mstep{c}{\ret{v}} \to 
    (\Sigma \code{v} : \kw{Pg}(\Sigma, A).\; (\code{e} \evalexp^{\eta_{\Cl} c} \code{v}) \times \code{v} \mathrel{R} v)}
\end{align*}
We lift the relation to closing instantiations of contexts $\kw{Inst}(\Sigma, \Gamma) \coloneqq \kw{Sub}_\Sigma(\Gamma, \kw{nil})$: 
\begin{align*}
  \textbf{Inductive}\; &{\approx} : \impl{\Sigma, \Gamma} \kw{Inst}(\Sigma, \Gamma) \to \semconMA{\Gamma}(\Sigma) \to \tpv\; \textbf{where}\\ 
    \kw{emp} &: \kw{emp} \approx_{\Sigma, \kw{nil}} \kw{\triv}\\ 
    \kw{cons} &: \impl{\Sigma, \Gamma, \code{\gamma}, \gamma, A, \code{a}, a} 
    \code{a} \approx_{\Sigma, A} a \to \code{\gamma} \approx_{\Sigma, \Gamma} \gamma \to 
    \kw{cons}(\code{a}, \code{\gamma}) \approx_{\Sigma, A::\Gamma} (a, \gamma)
\end{align*}

\subsubsection{Logical relation for commands}\label{sec:logical-cmd}

The logical relation for commands is obtained by instantiating the prelogical relation 
with the logical relation for expressions: 
\begin{align*}  
  {\sim} &: \impl{\Sigma, A}  
  \kw{Cmd}(\Sigma, A) \to (\semsigMA{\Sigma} \to \F{\semtyMA{A}(\Sigma) \times \semsigMA{\Sigma}}) \to \tpv\\ 
  m \sim_{\Sigma, A} c &= m \mathrel{\kw{cmd}_{\Sigma, A}(\approx_{\Sigma, A})} c
\end{align*}

\subsubsection{Fundamental theorem of logical relations for adequacy}\label{sec:ftlr-MA}

Using the axioms governing iteration and cost bounds introduced in \cref{sec:iteration}, w
we may prove the fundamental theorem of the logical relation by 
mutal induction on the derivation of expressions and commands. 
The details of the proof can be found in \cref{appendix:MA}.

\begin{theorem}[FTLR]\label{thm:ftlr-MA}
  Given an expression $e : \TmMA(\Sigma, \Gamma, A)$, if $p : \Sigma' \ge \Sigma$ 
  and $\code{\gamma'} \approx_{\Sigma', \Gamma} \gamma'$, then 
  $(\Uparrow^p e)[\code{\gamma'}] \approx^\Downarrow_{\Sigma', A} \semexpMA{e}(\Sigma', p, \gamma')$. 
  Moreover, given a command $m : \Cmd(\Sigma, \Gamma, A)$, if $p : \Sigma' \ge \Sigma$ 
  and $\code{\gamma'} \approx_{\Sigma', \Gamma} \gamma'$, then 
  $(\Uparrow^p m)[\code{\gamma'}] \sim_{\Sigma', A} \semcmdMA{m}(\Sigma', p, \gamma')$. 
\end{theorem}

As a corollary, we obtain cost-aware computational adequacy for both 
expressions and commands: 
\begin{corollary}[Cost-aware adequacy for \MA{}]
  Let $e : \TmMA(\cdot, \cdot, \code{bool})$ be a closed boolean with no free assignables. 
  If $\semexpMA{e} = \mstep{c}{\ret{b}}$, then we have $e \evalexp^{\eta_{\Cl} c} \overline{b}$.  
  Moreover, let $m : \Cmd(\cdot, \cdot, \code{bool})$ be a closed boolean command with no free assignables. 
  If $\semcmdMA{m} = \mstep{c}{\ret{(b, \triv)}}$, then we have $(\kw{emp}, m) \evalcmd^{\eta_{\Cl} c} (\kw{emp}, \overline{b})$.   
\end{corollary}

Extensional adequacy follows immediately: 
\begin{corollary}[Extensional adequacy for \MA{}]
  Suppose $u : \ExtOpn$. Let $e : \TmMA(\cdot, \cdot, \code{bool})$ be a closed boolean with no free assignables. 
  If $\semexpMA{e} = \ret{b}$, then we have $e \Downarrow \overline{b}$.  
  Moreover, let $m : \Cmd(\cdot, \cdot, \code{bool})$ be a closed boolean command with no free assignables. 
  If $\semcmdMA{m} = \ret{(b, \triv)}$, then we have $(\kw{emp}, m) \Downarrow_{\kw{cmd}} (\kw{emp}, \overline{b})$.   
\end{corollary}

%% file: models-calf.tex
\section{Models of \calf{}}

In this section we briefly recall from \citet{niu-sterling-grodin-harper:2022}
the notion of models of \calf{}, and in \cref{sec:model-calfpp,sec:metatheory-iter} 
we instantiate the parameters of this section with concrete constructions.  

\citet{niu-sterling-grodin-harper:2022} defines \calf{} as the free locally cartesian closed category 
generated by the associated signature (which we have present a fragment of in \cref{fig:calf}). 
A model of \calf{} in the sense of \opcit{} consists of any 
locally cartesian closed category $\CatIdent{E}$ and an implementation of the constants 
declared in the signature in $\CatIdent{E}$. 
\opcit{} constructs an Eilenberg-Moore-type model of 
\calf{} called the \emph{counting model} based on the writer 
monad associated to a cost monoid $\mathbb{C}$. 
In the next couple sections we extend the counting model to account for universes and partiality. 

%% file: model.tex
\section{A model of \calfpp{}}\label{sec:model-calfpp}

Fix a presheaf topos $\TopIdent{X}$.
We will construct a model \calfpp{} using the internal language of \TopIdent{X} ---
an extensional type theory equipped with (quotient) inductive types and a strict cumulative hierarchy of universes. 
As part of the input of the model construction, we are given a 
distinguished proposition $\kw{E} : \Omega$ 
in \TopIdent{X} representing the extensional phase. 
Following the notation of \citet{niu-sterling-grodin-harper:2022}, we write 
$\Op$ and $\Cl$ for the open and closed modalities associated with $\kw{E}$. 
In the following, let $\alpha < \beta < \gamma$ be universe levels 
and $\mathbb{C}$ be a cost monoid in the sense of \opcit{}. 
Recall that the counting model of \calf{} is based on the Eilenberg-Moore adjunction arising from 
the monad for cost effect $T \coloneqq \Cl \mathbb{C} \times -$, which we adopt for \calfpp{}:  

  \begin{minipage}{0.3\textwidth}  
  \begin{align*}
   \tpv &: \mathcal{U}_\gamma\\ 
   \tpv &= \mathcal{U}_\beta\\ 
   \tmv{A} & = A\\
  \end{align*}
  \end{minipage}
  \begin{minipage}{0.7\textwidth}  
    \begin{align*}
     \tpc &: \mathcal{U}_\gamma\\ 
     \tpc &= \kw{Alg}_{\mathcal{U}_\beta}(T)\\
     \tmc{X} &= |X|\\
    \end{align*}
    \end{minipage}
  Following the notation of \opcit{}, we write $\kw{Alg}_{\mathcal{U}}(T)$ for the type of 
  algebras for the monad $T$ whose carrier is valued in $\mathcal{U}$. 
  Given an algebra $\alpha : \kw{Alg}_{\mathcal{U}}(T)$, we write 
  $|\alpha|$ for the carrier and $\alpha.\kw{map}$ for the structure map.  
  The value universe can then be modeled as the universe $\mathcal{U}_\alpha : \mathcal{U}_\beta$, and the 
  computation universe is the \emph{trivial algebra}
  \footnote{The trivial algebra for the a writer monad $X \times -$ valued in $A$ is given 
  by the projection map $X \times A \to A$}  
  for $T$ valued in $\mathcal{U}_\alpha$-small 
  $T$-algebras: 

  \begin{minipage}{0.5\textwidth}  
    \begin{align*}
     \kw{Univ}^+ &: \mathcal{U}_\beta\\ 
     \kw{Univ}^+ &= \mathcal{U}_\alpha\\
    \end{align*}
    \end{minipage}
    \begin{minipage}{0.5\textwidth}  
      \begin{align*}
       \kw{Univ}^\ominus &: \kw{Alg}_{\mathcal{U}_\beta}(T)\\ 
       \kw{Univ}^\ominus &= \kw{trivAlg}(T, \kw{Alg}_{\mathcal{U}_\alpha}(T))\\
    \end{align*}
    \end{minipage}
  The decoding map for the value universe is the identity: $\kw{El}^+(A) = A$. 
  For the computation universe, we can unfold definitions and see that it suffices to 
  define a map $\kw{Alg}_{\mathcal{U}_\alpha}(T) \to \kw{Alg}_{\mathcal{U}_\beta}(T)$, 
  which is given by the inclusion of $\mathcal{U}_\alpha$-small algebras in 
  $\mathcal{U}_\beta$-small algebras. 

  \paragraph{Closure under type connectives}

  One may check that the pair of universes are closed under the connectives of \calfpp{}. 
  For instance, closure under dependent products requires implementations for the following: 
  \begin{align*}
    \widehat{\Pi} &: (A : \mathcal{U}_\alpha, B : A \to \kw{Alg}_{\mathcal{U}_\alpha}(T)) \to \kw{Alg}_{\mathcal{U}_\alpha}(T)\\ 
    \widehat{\Pi}/\kw{decode} &: (\widehat{\Pi}(A, B)) = \Pi(A, \lambda a : A.\; (B(a))) 
  \end{align*}
  We may implement $\widehat{\Pi}$ in the same way as $\Pi$ (except one universe level lower); 
  the decoding equation holds because \TopIdent{X} supports a strict cumulative universe hierarchy. 

  \paragraph{Inductive types}

  $W$-types exist in any topos with a natural numbers object. In particular, 
  this means we may interpret the internal $W$-type of \calfpp{} in any presheaf topos.

%% file: cbpv-models.tex
\section{A Model of \protect\calfiter{}}\label{sec:metatheory-iter}

In this section we extend the construction from \cref{sec:model-calfpp} to a model for \calfiter{}. 
For brevity, we have suppressed the verification of the axioms; the details may be 
found in \cref{appendix:model}. 

\paragraph{Notation}

Given a monad $M$, we write $\eta_M$, $\mu_M$ for the unit and multiplication of $M$, and 
we write $\kw{bind}_M$ or $\leftarrow_M$ for the derived bind operation. 

\subsection{Lifted computations}

The type of lifted computations $\Lift{A}$ are interpreted using the 
the \emph{quotient inductive-inductive partiality monad} of \citeauthor{altenkirch-danielsson-kraus:2017}
(written as $A_\bot$): 
\begin{align*}
  \kw{L} &: \mathcal{U}_\beta \to \kw{Alg}_{\mathcal{U}_\beta}(T)\\ 
  \Lift{A} &= \alpha_{\kw{L}(A)}
\end{align*}
The algebra for lifted computations is defined as follows:
\begin{align*}
  |\alpha_{\kw{L}(A)}| &= (T A)_\bot = (\Cl \mathbb{C} \times A)_\bot\\ 
  \alpha_{\kw{L}(A)} &: T|\alpha_{\kw{L}(A)}| \to |\alpha_{\kw{L}(A)}|\\
  \alpha_{\kw{L}(A)}(c, e) &= (c', a) \leftarrow_\bot e; \eta_\bot(c + c', a)
\end{align*}
It is straightforward to verify that the algebra laws are satisfied.
The inclusion of free computations in lifted computations is given by $\eta_\bot$: 
\begin{align*}
  \kw{ret}_{\kw{L}} &: A \to (\Cl \mathbb{C} \times A)_\bot\\ 
  \retl{a} &= \eta_\bot(\eta_{\Cl}0, a) 
\end{align*}
Sequencing of lifted computations is implemented by threading through the cost 
of computations: 
\begin{align*}
  \kw{bind}_{\kw{L}} &: (\Cl \mathbb{C} \times A)_\bot \to (A \to (\Cl \mathbb{C} \times B)_\bot) \to (\Cl \mathbb{C} \times B)_\bot\\
  \kw{bind}_{\kw{L}}(e, f) &= (c_1, a) \leftarrow_\bot e; (c_2, b) \leftarrow_\bot f a; \eta_\bot(c_1 + c_2, b) 
\end{align*}

%% file: conclusion.tex
\section{Conclusion}\label{sec:conclusion}

Denotational semantics is a well-established method for 
studying the \emph{extensional} property of programs. 
In this paper we contribute a family of  
\emph{cost-aware} metalanguages for studying \emph{intensional} properties via 
\emph{synthetic} denotational semantics. 
The metalanguage we present supports synthetic reasoning in two orthogonal directions. 
First, by basing our work on \calf{}, a dependent type theory with an axiomatic theory of 
the interaction of intension and extension \citep{niu-sterling-grodin-harper:2022}, 
we obtain a rich language for phase-separated constructions. 
As we show in \cref{sec:adequacy-STLC,sec:adequacy-MA}, this enables us to 
formulate and prove cost-aware generalizations of classic Plotkin-type adequacy theorems 
that restrict immediately to their original extensional counterparts, which improves upon prior work on 
synthetic denotational semantics in type theory (see \cref{sec:den-sem-guarded}). 
Second, our metalanguage is also synthetic in a more traditional sense by allowing the user to 
construct conceptually simple denotational semantics of programming languages using only elementary 
type-theoretic constructions. 

We illustrate our approach by proving a cost-aware computational adequacy theorem in the style of Plotkin for 
the simply-typed lambda calculus and Modernized Algol. 
These results establish criterions by which cost models for algorithm analysis in \calf{} may be deemed to be 
cost adequate with respect to a given operational semantics, thereby  
giving a positive answer to the conjecture of \citet{niu-sterling-grodin-harper:2022}. 
In view of \opcit{}'s work on algorithm analysis, the metalanguage we have developed 
constitutes an expressive framework for not only cost-aware programming and verification 
but also cost-aware metatheory of programming languages. 

\paragraph{Future work}

In \cref{sec:SDT} we mentioned the possibility of extending our work to account for PCF and 
truly generalizing Plotkin's original adequacy theorem. The main challenge would be to 
construct a suitable topoi with an SDT theory to obtain a 
domain that supports \emph{arbitrary} fixed-points on endofunctions. 
We believe recent work on the topos-theoretic development of programming language metatheory using SDT 
\citep{sterling-harper:2022} will be germane.  

More generally, as a burgeoning area of research, cost-aware programming and synthetic metatheory is ripe 
with questions and challenges. In \cref{sec:compiler} we relate our work to the field of 
compiler correctness, but there are many other opportunities. 
By developing a framework for synthetic cost-aware denotational semantics, we hope to 
build the groundworks for more investigations of classic ideas from a fresh, synthetic perspective 
that may shed light on old and new problems alike.

%% file: appendix.tex
\section{Cost-aware adequacy proof for \STLC{}}\label{appendix:STLC}

\begin{lemma}[Hypothesis]\label{lemma:hyp}
  If $v : \kw{Var}(\Gamma, A)$ and $\code{\gamma} \approx_{\Gamma} \gamma$, then 
  $v[\code{\gamma}] \approx_A \semvar{v}(\gamma)$. 
\end{lemma}

\begin{proof}
  By induction on derivation of $v : \kw{Var}(\Delta, A)$. 
  \paragraph{\textbf{Case}: $v = \kw{now}$}
  By assumption, we know that there exist 
  $\code{a}$ and $a$ such that $\code{\gamma} = \kw{cons}(\code{a}, \code{\gamma'})$,
  $\gamma = (a, \gamma')$, $\code{a} \approx_A a$, and $\code{\gamma'} \approx_{\Gamma} \gamma'$. 
  We want to show $\code{a} \approx_A a$, which is immediate by assumption. 
  \paragraph{\textbf{Case}: $v = \kw{next}(v')$ for some $v' : \kw{Var}(\Gamma, A)$}
  By inductive hypothesis. 
\end{proof}

\begin{theorem}[FTLR]
  Given a STLC term $e : \TmSTLC(\Gamma, A)$, if $\code{\gamma} \approx_{\Gamma} \gamma$, then 
  $e[\code{\gamma}] \approx^\Downarrow_A \semtm{e}(\gamma)$. 
\end{theorem}

\begin{proof}
  Induction on the derivation of the term $e : \TmSTLC(\Gamma, A)$. 
  \paragraph{Variable} 
  By \cref{lemma:hyp}. 
  \paragraph{\textbf{Case}: observables} 
  WLOG, suppose that $e = \code{tt}$. We want to show that 
  $\code{tt} \approx^{\downarrow}_{\code{bool}} \ret{\kw{tt}}$. 
  First, note that we have $\code{tt} \eval^{\eta_{\Cl} 0} \code{tt}$ and 
  $\ret{\kw{tt}} = \mstep{0}{\ret{\kw{tt}}}$, so it suffices to show 
  $\code{tt} \approx_{\code{bool}} \kw{tt}$, which holds
  since $\overline{\kw{tt}} = \code{tt}$. 
  \paragraph{\textbf{Case}: functions} 
  We have to show $(\code{lam}(\code{e})[\code{\gamma}]) \approx^{\Downarrow}_{A_1 \Rightarrow A_2} \semtm{\code{lam}(\code{e})}(\gamma)$. 
  Unfolding definitions, it suffices to show that 
  $\code{lam}(\code{e}[\code{\gamma}\uparrow^{A_1}]) \approx_{A_1 \Rightarrow A_2} \lambda a : \semty{A_1}.\; \semtm{\code{e}}(a, \gamma)$. 
  Suppose that $\code{e_1} \approx_{A_1} e_1$. We have to show 
  $\code{e}[\code{\gamma} \uparrow^{A_1}][\code{e_1}] \approx^{\Downarrow}_{A_2} \semtm{\code{e}}(e_1, \gamma)$. 
  By \cref{prop:subst}, we have $\code{e}[\code{\gamma} \uparrow^{A_1}][\code{e_1}] = \code{e}[\kw{cons}(\code{e_1}, \code{\gamma})]$, 
  so the result would follow if we can show that $\kw{cons}(\code{e_1}, \code{\gamma}) \approx_{A_1::\Gamma} (e_1, \gamma)$. 
  This follows from $\approx.\kw{cons}$ and assumptions.   
  \paragraph{\textbf{Case}: application}  
  We have to show that $(\code{ap}(e, e_1))[\code{\gamma}] \approx^{\Downarrow}_{A_2} \semtmSTLC{\code{ap}(e, e_1)}(\gamma)$. 
  By induction on $e$, we have that 
  $e \eval^{\eta_{\Cl} c} \code{f}$ for some $\code{f} : \kw{Pg}(A_1 \Rightarrow A_2)$ 
  and $c : \Nat$, $\semtmSTLC{e} = \mstep{c}{\ret{f}}$ for some $f : \semtySTLC{A_1} \to \F{\semtySTLC{A_2}}$, and 
  $\code{f} \approx_{A_1 \Rightarrow A_2} f$. 
  Computing the definition of the logical relation at $A_1 \Rightarrow A_2$, we have that 
  $\code{f} = \code{lam}(\code{e_2})$ for some $\code{e_2}$ and a term $h$ of the following type: 
  \[
    ((\code{e_1} : \kw{Pg}(A_1), e_1 : \semty{A_1}) \to \code{e_1} \approx_{A_1} e_1
    \to \code{e_2}[\code{e_1}] \approx^\Downarrow_{A_2} e(e_1)) 
  \]
  By induction on $e_1$, we have that 
  $e_1 \eval^{\eta_{\Cl} c_1} \code{v_1}$ for some $\code{v_1} : \kw{Pg}(A_1)$ 
  and $c_1 : \Nat$, $\semtmSTLC{e_1} = \mstep{c_1}{\ret{v_1}}$ for some $v_1 : \semtySTLC{A_1}$, and 
  $\code{v_1} \approx_{A_1} v_1$. 
  Instantiating $h$, we have that $\code{e_2}[\code{v_1}] \approx^{\Downarrow}_{A_2} f(v_1)$, 
  which means that there exists $c_2$, $\code{v}$, and $v$ such $\code{e_2}[\code{v_1}] \eval^{\eta_{\Cl} c_2} \code{v}$, 
  $f(v_1) = \mstep{c_2}{\ret{v}}$, and $\code{v} \approx_{A_2} v$.
  Combined with the fact that $e \eval^{\eta_{\Cl} c} \code{f}$ and $e_1 \eval^{\eta_{\Cl} c_1} \code{v_1}$,  
  we have that $\code{ap}(e, e_1) \eval^{\eta_{\Cl}(c + c_1 + 1 + c_2)} \code{v}$. 
  Moreover, we can compute the meaning of $\code{ap}(e, e_1)$: 
  \begin{align*}
    \semtmSTLC{\code{ap}(e, e_1)}(\gamma) &=  
    \bind{\semtm{e}(\gamma)}{\lambda f.\; \bind{\semtm{e_1}(\gamma)}{\lambda a.\; \mstep{1}{f(a)}}}\\
    &= \bind{\mstep{c}{\ret{f}}}{\lambda f.\; \bind{\mstep{c_1}{\ret{v_1}}}{\lambda a.\; \mstep{1}{f(a)}}}\\
    &= \mstep{c}{\bind{\mstep{c_1}{\ret{v_1}}}{\lambda a.\; \mstep{1}{f(a)}}}\\
    &= \mstep{c}{\mstep{c_1}{\mstep{1}{f(v_1)}}}\\
    &= \mstep{c}{\mstep{c_1}{\mstep{1}{\mstep{c_2}{\ret{v}}}}}\\
    &= \mstep{c+c_1+1+c_2}{\ret{v}}
  \end{align*}
  And this is what we needed to show.
\end{proof}

\section{Cost-aware adequacy proof for \MA{}}\label{appendix:MA}

\subsection{Properties of substitution}

\begin{proposition}
  There is a map $\kw{sh} : \impl{\Sigma, \Sigma'} \Sigma' \ge \Sigma \to \Nat \to \Nat$  
\end{proposition}

\begin{proof}
  Define $\kw{sh}$ as follows:
  \begin{align*}
    \kw{sh} &: \impl{\Sigma, \Sigma'} \Sigma' \ge \Sigma \to \Nat \to \Nat\\
    \kw{sh}(\kw{refl}, n) &= n\\
    \kw{sh}(\kw{mono}(p), 0) &= 0\\
    \kw{sh}(\kw{mono}(p), n+1) &= \kw{sh}(p, n) + 1\\
    \kw{sh}(\kw{extend}(p)) &= \kw{sh}(p, n) + 1
  \end{align*}
\end{proof}

\begin{proposition} 
  There is a map $\Uparrow : \impl{\Sigma, \Sigma', \Gamma, A} \Sigma' \ge \Sigma \to \Cmd(\Sigma, \Gamma, A)\to \Cmd(\Sigma', \Gamma, A)$.  
\end{proposition}

\begin{proof}
  We need to define the map mutual recursively:  

  \begin{minipage}{0.45\textwidth}
    \begin{align*}
      \Uparrow^p (\code{ret}(a)) &= \code{ret}(\Uparrow^p a)\\
      \Uparrow^p (\code{bnd}(e, m)) &= \code{bnd}(\Uparrow^p e, \Uparrow^p m)\\ 
      \Uparrow^p (\code{while}[n](m)) &= \code{while}[\kw{sh}(p, n)](\Uparrow^p m)\\ 
      \Uparrow^p (\code{get}[n]) &= \code{get}[\kw{sh}(p, n)]\\ 
      \Uparrow^p (\code{set}[n](e)) &= \code{set}[\kw{sh}(p, n)](\Uparrow^p m)\\ 
      \Uparrow^p (\code{dcl}(e, m)) &= \code{dcl}(\Uparrow^p e, \Uparrow^{\kw{mono}(p)} m)
    \end{align*}
  \end{minipage} 
  \begin{minipage}{0.45\textwidth}
    \begin{align*} 
      \Uparrow^p (\code{lam}(e)) &= \code{lam}(\Uparrow^p e)\\ 
      \Uparrow^p (\code{ap}(e, e_1)) &= \code{ap}(\Uparrow^p e, \Uparrow^p e_1)\\ 
      \Uparrow^p (\code{suc}(e)) &= \code{suc}(\Uparrow^p e)\\ 
      \Uparrow^p (\code{ifz}(e, e_1, e_2)) &= \code{ifz}(\Uparrow^p e, \Uparrow^p e_1, \Uparrow^p e_2)\\ 
      \Uparrow^p (\code{cmd}(m)) &= \code{cmd}(\Uparrow^p m)\\
      \Uparrow^p e &= e 
    \end{align*} 
  \end{minipage}
\end{proof}

One may weaken a substitution, which we write as
${\uparrow} : \impl{\Sigma, \Gamma, \Gamma'} (A : \TyMA) \to \kw{Sub}_\Sigma(\Gamma, \Gamma') \to \kw{Sub}_\Sigma(A::\Gamma, A::\Gamma')$.
The weakening of signatures may be commuted past a substitution:  
\begin{proposition}
  Given $e : \TmMA(\Sigma, \Gamma, A)$, $p : \Sigma' \ge \Sigma$, 
  and $\gamma : \kw{Sub}_\Sigma(\Gamma, \Gamma')$, 
  we have that $\Uparrow^p (e[\gamma]) = (\Uparrow^p e)[\Uparrow^p \gamma]$. 
\end{proposition}
Conversely, Weakening of substitutions may be commuted past weakening of signatures: 
\begin{proposition}
  Given $\gamma : \kw{Sub}_\Sigma(\Gamma, \Gamma')$ and $p : \Sigma' \ge \Sigma$, we have that 
  $\uparrow^A (\Uparrow^p \gamma)  = \Uparrow^p (\uparrow^A\gamma)$. 
\end{proposition}

\begin{proposition}
  Given $e : \TmMA(\Sigma, \Gamma, A)$, 
  we have that $\Uparrow^p (\Uparrow^q e) = \Uparrow^{\kw{tr}(p, q)} e$; moreover, 
  given $m : \Cmd(\Sigma, \Gamma, A)$,   
  we have that $\Uparrow^p (\Uparrow^q m) = \Uparrow^{\kw{tr}(p, q)} m$. 
\end{proposition}

Moreover, the analog to \cref{prop:subst} holds for substitution as defined in \cref{def:subst-MA}: 
\begin{proposition} \label{prop:subst-MA}
  Given $e : \TmMA(\Sigma, A::\Gamma, A')$, $\sigma : \kw{Sub}_\Sigma(\Gamma, \kw{nil})$, 
  and $e' : \kw{Tm}(\Sigma, \kw{nil}, A)$, we have that 
  $e[\uparrow^A\sigma][e'] = e[\kw{cons}(e', \sigma)]$. 
  Moreover, given $m : \Cmd(\Sigma, A::\Gamma, A')$, $\sigma : \kw{Sub}_\Sigma(\Gamma, \kw{nil})$, 
  and $e' : \kw{Tm}(\Sigma, \kw{nil}, A)$, we have that 
  $m[\uparrow^A\sigma][e'] = m[\kw{cons}(e', \sigma)]$. 
\end{proposition}

\subsection{Properties of phase-separated evaluation}

\begin{proposition}\label{prop:eval-dcl}
  Let $(A, h_A) : \kw{Pos}$ and $(B, h_B) : \kw{Pos}$ be positive types, and let 
  $e : \TmMA(\Sigma, A)$ and $m : \Cmd(A::\Sigma, B)$. 
  If $e \Downarrow^{c_1} a$ for some $c_1 : \Nat$, $a : \kw{Val}(\cdot, A)$, and 
  $(a::\mu, m) \Downarrow^{c_2} (-::\mu', \code{ret}(b))$ for some $c_2 : \Nat$ and $\mu'$, then 
  $(\mu, \code{dcl}(e, m)) \Downarrow^{c_1 + c_2 + \eta_{\Cl} 1}_{\kw{cmd}} (\mu', \code{ret}(\kw{coer}(b)))$. 
\end{proposition}

\begin{proposition}\label{prop:eval-bnd}
  Let $A, B : \TyMA$ be types, and let $e : \TmMA(\Sigma, \code{cmd}(A))$ and $m : \Cmd(\Sigma, A, B)$. 
  If $e \evalexp^{c_1} \code{cmd}(m_1)$, $(\mu, m_1) \Downarrow^{c_2}_{\kw{cmd}} (\mu_1, \code{ret}(a))$, 
  and $(\mu, m[a]) \Downarrow^{c_3} (\mu', \code{ret}(b))$, 
  then we have that $(\mu, \code{bnd}(e, m)) \evalcmd^{c_1 + c_2 + c_3 + \eta_{\Cl}(1)} (\mu', \code{ret}(b))$.
\end{proposition}

\begin{proposition}\label{prop:eval-set}
  Let $(A, h_A) :\kw{Pos}$, $e : \TmMA(\Sigma, A)$, and $\Sigma[n] = (A, h_A)$. 
  If $e \Downarrow^{c} a$, then $(\mu, \code{set}[n](e)) \Downarrow^{c+\eta_{\Cl}1} (\mu'[n \mapsto a], \mu'[n])$. 
\end{proposition}

\subsection{Properties of the denotational semantics}

For strictly positive types, one may transport an expression between 
arbitrary signatures:
\begin{proposition}
  For all $\Sigma, \Sigma' : \kw{Sig}$, if $A$ is a positive type and $a : \TmMA(\Sigma, \Gamma, A)$, 
  then there is a term $\kw{coer}(a) : \TmMA(\Sigma', \Gamma, A)$.  
\end{proposition}

We have that the interpretaion of types is independent of signatures on strictly positive types: 
\begin{proposition}\label{prop:tm-pos}
  If $(A, h_A) : \kw{Pos}$ then $\semtyMA{A}(\Sigma) = \semexpMA{A}(\Sigma')$ for all 
  $\Sigma, \Sigma' : \kw{Sig}$. 
\end{proposition}

\subsection{Proof of \cref{thm:ftlr-MA}}

\begin{proposition}\label{prop:pos-rel}
  If $(A, h_A) : \kw{Pos}$ and $\code{a} \approx_{\Sigma} a$, then $\kw{coer}(\code{a}) \approx_{\Sigma'} a$ for all 
  $\Sigma, \Sigma' : \kw{Sig}$.
\end{proposition}

\begin{lemma}[Hypothesis]\label{lemma:hyp-MA}
  If $v : \kw{Var}(\Gamma, A)$ and $\code{\gamma} \approx_{\Sigma, \Gamma} \gamma$, then 
  $v[\code{\gamma}] \approx_{\Sigma, A} \semvar{v}(\gamma)$. 
\end{lemma}

\begin{theorem}[FTLR]
  Given an expression $e : \TmMA(\Sigma, \Gamma, A)$, if $p : \Sigma' \ge \Sigma$ 
  and $\code{\gamma'} \approx_{\Sigma', \Gamma} \gamma'$, then 
  $(\Uparrow^p e)[\code{\gamma'}] \approx^\Downarrow_{\Sigma', A} \semexpMA{e}(\Sigma', p, \gamma')$. 
  Moreover, given a command $m : \Cmd(\Sigma, \Gamma, A)$, if $p : \Sigma' \ge \Sigma$ 
  and $\code{\gamma'} \approx_{\Sigma', \Gamma} \gamma'$, then 
  $(\Uparrow^p e)[\code{\gamma'}] \sim_{\Sigma', A} \semcmdMA{m}(\Sigma', p, \gamma')$. 
\end{theorem}

\subsubsection{\textbf{Case}: iteration}
  By assumption, we have that $\Sigma[n] = (\code{bool}, -)$ and 
  $m : \Cmd(\Sigma, \Gamma, \code{unit})$. 
  Suppose that $p : \Sigma' \ge \Sigma$ and $\code{\gamma'} \approx_{\Sigma', \Gamma} \gamma'$. 
  We have to show that $(\Uparrow^p \code{while}[n](m)) [\code{\gamma'}] \sim_{\Sigma', \code{unit}} \semcmdMA{\code{while}[n](m)}(\Sigma', p, \gamma')$. 
  We may compute each side: 
  \begin{align*}
    (\Uparrow^p \code{while}[n](m)) [\code{\gamma'}] &= \code{while}[\kw{sh}(p, n)](\Uparrow^p m) [\code{\gamma'}] =\code{while}[\kw{sh}(p, n)]((\Uparrow^p m) [\code{\gamma'}])
  \end{align*}
  \begin{align*}
    \semcmdMA{\code{while}[n](m)}(\Sigma', p, \gamma') &= \lambda \sigma' : \semsigMA{\Sigma'}.\; \kw{iter}(g)(\sigma')
  \end{align*}
  
  Let $c : \Nat$, $\sigma, \sigma' : \semsigMA{\Sigma'}$. 
  Suppose that $\kw{iter}(g)(\sigma') = \mstep{c}{\retl{\triv, \sigma''}}$.  
  By $\kw{iter}/\kw{trunc}$, we know that 
  $\kw{seq}(g, k, \sigma') = \mstep{c}{\retl{\inl{\triv, \sigma''}}}$ for some $k : \Nat$. 
  We have to show that 
  if $\mu' \sim_{\Sigma'} \sigma'$ then 
  $(\mu', \code{while}[\kw{sh}(p, n)]((\Uparrow^p m) [\code{\gamma'}])) \evalcmd^{\eta_{\Cl} c}
  (\mu'', \code{ret}(\code{u}))$ such that $\code{u} \approx_{\Sigma', \code{unit}} \triv$ and $\mu'' \sim_{\Sigma'} \sigma''$.  
  We proceed by induction on $k$. 

  If $k = 0$, by definition we have $\kw{seq}(g, 0, \sigma') = \retl{\inr{\sigma'}}$, and so 
  $\mstep{c}{\retl{\inl{\triv, \sigma''}}} = \retl{\inr{\sigma'}}$ as well. 
  By \cref{prop:step-ret-lift}, we have that $\inl{\triv, \sigma''} = \inr{\sigma'}$, 
  which is a contradiction.  

  Otherwise, we have that $k = k' + 1$ for some $k' : \Nat$. 
  By definition, we have that $\kw{seq}(g, k' + 1, \sigma') = \bindl{g \sigma'}{[\kw{ret} \circ \kw{inl}; \kw{seq}(g, k')]}$. 
  We proceed by cases on $\sigma'[\kw{sh}(p, n)] : \kw{bool}$. 

  If $\sigma'[\kw{sh}(p, n)] = \kw{ff}$, then we have that 
  $g \sigma' = \mstep{1}{\retl{\inl{\triv, \sigma'}}}$, and so we have the following: 
  \begin{align*}
    \kw{seq}(g, k' + 1, \sigma') &= \bindl{g \sigma'}{[\kw{ret}_{\kw{L}} \circ \kw{inl}; \kw{seq}(g, k')]}\\ 
    &= \bindl{\mstep{1}{\retl{\inl{\triv, \sigma'}}}}{[\kw{ret}_{\kw{L}} \circ \kw{inl}; \kw{seq}(g, k')]}\\ 
    &= \mstep{1}{\retl{\inl{\triv, \sigma'}}}\\
    &= \mstep{c}{\retl{\inl{\triv, \sigma''}}}
  \end{align*}
  By $\kw{step}/\kw{inj}$, we have that $\sigma' = \sigma''$ and a term $h : \Cl(c = 1)$. 
  Now suppose that $\mu' \sim_{\Sigma'} \sigma'$. 
  Because $\sigma'[\kw{sh}(p, n)] = \kw{ff}$, we also know that $\mu'[\kw{sh}(p, n)] = \code{ff}$, 
  which means that $(\mu', \code{while}[\kw{sh}(p, n)]((\Uparrow^p m) [\code{\gamma'}])) \Mapsto 
  (\mu', \code{ret}(\code{\triv}))$. 
  By definition of the logical relation for expressions, we have $\code{\triv} \approx_{\Sigma', \code{unit}} \triv$, 
  and combined with the assumption we have $\mu' \sim_{\Sigma'} \sigma'$.
  Consequently, it suffices to show that 
  $(\mu', \code{while}[\kw{sh}(p, n)]((\Uparrow^p m) [\code{\gamma'}])) \evalcmd^{\eta_{\Cl} 1} (\mu', \code{ret}(\code{\triv}))$,
  which follows from the definition of phase-separated evaluation.  
  
  Otherwise, we have that $\sigma'[\kw{sh}(p, n)] = \kw{tt}$. By definition,
  we have the following: 
  \[g \sigma' = (-, \sigma_1) \leftarrow_{\kw{L}} \semcmdMA{m}(\Sigma', p, \gamma', \sigma'); \mstep{1}{\retl{\inr{\sigma_1}}}\]
  and so we have the following:
  \begin{align*}
    \kw{seq}(g, k' + 1, \sigma') &= (-, \sigma_1) \leftarrow_{\kw{L}} \semcmdMA{m}(\Sigma', p, \gamma', \sigma'); \mstep{2}{\kw{seq}(g, k', \sigma_1)}\\
    &= \mstep{c}{\retl{\inl{\triv, \sigma''}}}
  \end{align*}
  By $\kw{bind}_{\kw{L}}^{-1}$ and $\kw{step}_{\kw{L}}^{-1}$, we have that 
  $\semcmdMA{m}(\Sigma', p, \gamma', \sigma') = \mstep{c_1}{\retl{-, \sigma_1}}$ for some $\sigma_1 : \semsigMA{\Sigma'}$ and 
  $c_1$ and $\kw{seq}(g, k', \sigma_1) = \mstep{c_2}{\retl{\inl{\triv, \sigma''}}}$ 
  for some $c_2$ such that and $h : \Cl(c = c_1 + c_2 + 2)$. 
  Now suppose that $\mu' \sim_{\Sigma'} \sigma'$.  
  By induction on $m$, we have that $(\Uparrow^p m)[\code{\gamma'}] \sim_{\Sigma', \code{unit}} \semcmdMA{m}(\Sigma', p, \gamma')$. 
  Unfolding the definition of the logical relation for commands, we have that 
  $(\mu', (\Uparrow^p m)[\code{\gamma'}]) \evalcmd^{\eta_{\Cl} c_1} (\mu_1, \code{ret}(\code{u}))$ for some 
  $\code{u}$ and $\mu_1$ such that $\code{u} \approx_{\Sigma', \code{unit}} \triv$ and 
  $\mu_1 \sim_{\Sigma'} \sigma_1$. Moreover, by the induction hypothesis on $k'$, we have that 
  $(\mu_1, \code{while}[\kw{sh}(p, n)]((\Uparrow^p m) [\code{\gamma'}])) \evalcmd^{\eta_{\Cl} c_2}
  (\mu'', \code{ret}(\code{u'}))$ for some $\code{u'}$ and $\mu''$ such that $\code{u'} \approx_{\Sigma', \code{unit}} \triv$ and $\mu'' \sim_{\Sigma'} \sigma''$.  
  By \cref{prop:eval-bnd}, we have the following:
  \[(\mu', \code{bnd}(\code{cmd}((\Uparrow^p m) [\code{\gamma'}]), \code{while}[\kw{sh}(p, n)]((\Uparrow^p m) [\code{\gamma'}])))
  \Downarrow^{\eta_{\Cl}(c_1 + c_2 + 1)} (\mu'', \code{ret}(\code{u'}))\] 
  Since $\mu' \sim_{\Sigma'} \sigma'$ and $\sigma'[\kw{sh}(p, n)] = \kw{tt}$, we know that 
  $\mu'[\kw{sh}(p, n)] = \code{tt}$, so we also the following:
  \[(\mu', \code{while}[\kw{sh}(p, n)]((\Uparrow^p m) [\code{\gamma'}])) \Mapsto
  (\mu', \code{bnd}(\code{cmd}((\Uparrow^p m) [\code{\gamma'}]), \code{while}[\kw{sh}(p, n)]((\Uparrow^p m) [\code{\gamma'}])))\] 
  Therefore, we have the following:
  \[(\mu', \code{while}[\kw{sh}(p, n)]((\Uparrow^p m) [\code{\gamma'}]))
  \Downarrow^{\eta_{\Cl}(c_1 + c_2 + 2)} (\mu'', \code{ret}(\code{u'}))\] 
  The result then follows from $h : \Cl(c = c_1 + c_2 + 2)$, as 
  it implies $\eta_{\Cl} c = \eta_{\Cl} (c_1 + c_2 + 2)$.  

\subsubsection{\textbf{Case}: sequence} 
  Suppose $p : \Sigma' \ge \Sigma$ and $\code{\gamma'} \approx_{\Sigma', \Gamma} \gamma'$. 
  We have to show that $(\Uparrow^p \code{bnd}(e, m)) [\code{\gamma'}] \sim_{\Sigma', B} \semcmdMA{\code{bnd}(e,m)}(\Sigma', p, \gamma')$. 
  We compute: 

  \paragraph{Computation}
  \begin{align*}  
  (\Uparrow^p \code{bnd}(e, m)) [\code{\gamma'}] &= \code{bnd}(\Uparrow^p e, \Uparrow^p m)[\code{\gamma'}]\\ 
  &=\code{bnd}((\Uparrow^p e) [\code{\gamma'}], (\Uparrow^p m)[\uparrow^A\code{\gamma'}]) 
  \end{align*}

  \paragraph{Computation}
  \begin{align*}
    \semcmdMA{\code{bnd}(e,m)}(\Sigma', p, \gamma') &= \lambda \sigma'.\; \\
    m_1 &\leftarrow_{\kw{L}} \kw{lift}(\semexpMA{e}(\Sigma', p, \gamma'));\\
    (a, \sigma_1) &\leftarrow_{\kw{L}} m_1(\Sigma', \kw{refl}, \sigma');\\
    &\step{\semcmdMA{m}(\Sigma', p, (a, \gamma'), \sigma_1)}\\ 
  \end{align*}

  Let $\sigma', \sigma'' : \semsigMA{\Sigma'}$ and 
  suppose that $\semcmdMA{\code{bnd}(e,m)}(\Sigma', p, \gamma')(\sigma') = \mstep{c}{\retl{b, \sigma''}}$
  for some $c : \Nat$ and $b : \semtyMA{B}(\Sigma')$. 
  Moreover, suppose that $\mu' : \kw{Store}(\Sigma')$ such that $\mu' \sim_{\Sigma'} \sigma'$. 
  We have to show that $(\mu', \code{bnd}((\Uparrow^p e) [\code{\gamma'}], (\Uparrow^p m)[\uparrow^A\code{\gamma'}]) ) 
  \Downarrow^{\eta_{\Cl} c}_{\kw{cmd}} (\mu'', \code{ret}(\code{b}))$ for some 
  $\code{b} \approx_{\Sigma', B} b$ and $\mu'' \sim_{\Sigma'} \sigma''$. 
  By the assumption that $\semcmdMA{\code{bnd}(e,m)}(\Sigma', p, \gamma')(\sigma')$ 
  has a normal form, we may apply 
  $\kw{bind}_{\kw{L}}^{-1}$ and $\kw{step}_{\kw{L}}^{-1}$ to obtain the following: 
  \begin{enumerate}
    \item $\kw{lift}(\semexpMA{e}(\Sigma', p, \gamma')) = \mstep{c_1}{\retl{m_1}}$ for some 
    $c_1$ and $m_1 : \semtyMA{\code{cmd}(A)}(\Sigma')$. \label{hyp:e}
    \item $m_1(\Sigma', \kw{refl}, \sigma') = \mstep{c_2}{\retl{a, \sigma_1}}$ for some 
    $c_2$, $a : \semtyMA{A}(\Sigma')$ and $\sigma_1 : \semsigMA{\Sigma'}$. \label{hyp:m-1}
    \item $\semcmdMA{m}(\Sigma', p, (a, \gamma'), \sigma_1) = \mstep{c_3}{\retl{b, \sigma''}}$ for some 
    $c_3$. \label{hyp:m} 
    \item $\Cl (c = c_1 + c_2 + c_3 + 1)$. \label{hyp:eq-costs}
  \end{enumerate}
  By induction on $e$, we know that 
  $(\Uparrow^ p e) [\code{\gamma'}] \approx^{\Downarrow}_{\Sigma', \code{cmd}(A)} \semexpMA{e}(\Sigma', p, \gamma')$. 
  From \cref{hyp:e} and the injectivity of $\kw{lift}$ we know that 
  $\semexpMA{e}(\Sigma', p, \gamma') = \mstep{c_1}{\ret{m_1}}$, so 
  by definition of the lift of the logical relation $-^{\Downarrow}$, we know that 
  there exists a program $\code{v_1} : \kw{Pg}(\Sigma', \code{cmd}(A))$ such that 
  $(\Uparrow^ p e) [\code{\gamma'}] \Downarrow^{\eta_{\Cl} c_1} \code{v_1}$ and $\code{v_1} \approx_{\Sigma', \code{cmd}(A)} m_1$. 
  Because $\code{v_1}$ is a value, we know that $\code{v_1} = \code{cmd}(\code{m_1})$ for some 
  command $\code{m_1} : \kw{Cmd}(\Sigma', A)$. 
  By definition of the logical relation for expressions, this implies that 
  $\code{m_1} \sim_{\Sigma', A} m_1(\Sigma', \kw{refl})$ holds. 
  By \cref{hyp:m-1} and the definition of the relation for commands, there exists $\mu_1$ and $\code{a}$ such that
  $(\mu', \code{m_1}) \evalcmd^{\eta_{\Cl} c_2} (\mu_1, \code{ret}(\code{a}))$ 
  and $\code{a} \approx_{\Sigma', A} a$ and $\mu_1 \sim_{\Sigma'} \sigma_1$. 
  Along with the definition of the logical relation for contexts, we have that 
  $\kw{cons} (\code{a}, \code{\gamma'}) \approx_{\Sigma', A::\Gamma} (a, \gamma')$. 
  Finally, by induction on $m$, we have that 
  $(\Uparrow^p m) [\kw{cons} (\code{a}, \code{\gamma'})] \sim_{\Sigma', B} \semcmdMA{m}(\Sigma', p, (a, \gamma'))$. 
  By \cref{hyp:m} and the fact that $\mu_1 \sim_{\Sigma'} \sigma_1$, 
  there exists $\mu''$ and $\code{b}$ such that $(\mu_1, (\Uparrow^p m) [\kw{cons} (\code{a}, \code{\gamma'})]) 
  \evalcmd^{\eta_{\Cl} c_3} (\mu'', \code{ret}(\code{b}))$ and
  $\code{b} \approx_{\Sigma', B} b$ and $\mu'' \sim_{\Sigma'} \sigma''$. 
  
  By \cref{prop:subst-MA}, we have that $(\Uparrow^p m) [\kw{cons} (\code{a}, \code{\gamma'})] = (\Uparrow^p m)[\uparrow^A\code{\gamma'}][\code{a}]$, 
  so the result will follow by \cref{prop:eval-bnd}, given that 
  we can show $\eta_{\Cl} c = \eta_{\Cl} c_1 + \eta_{\Cl} c_2 + \eta_{\Cl} c_3 + \eta_{\Cl} 1$. 
  Since this holds by \cref{hyp:eq-costs}, we are done. 

\subsubsection{\textbf{Case}: allocation}
  By assumption we have that $(A, h_A) : \kw{Pos}$ and $(B, h_B) : \kw{Pos}$, 
  $e : \TmMA(\Sigma, \Gamma, A)$, and $m : \Cmd(A::\Sigma, \Gamma, B)$. 
  Suppose $p : \Sigma' \ge \Sigma$ and $\code{\gamma'} \approx_{\Sigma', \Gamma} \gamma'$.  
  We have to show that $(\Uparrow^p (\code{dcl}(e, m)))[\code{\gamma'}] \sim_{\Sigma', B} 
  \semcmdMA{\code{dcl}(e, m)}(\Sigma', p, \gamma')$. 
 
  We may compute each side of this relation:
  \paragraph{Computation}
  \begin{align*}
    (\Uparrow^p (\code{dcl}(e, m)))[\code{\gamma'}] &= \code{dcl}(\Uparrow^p e, \Uparrow^{\kw{mono}(p)} m)[\code{\gamma'}] \\
    &= \code{dcl}((\Uparrow^p e) [\code{\gamma'}], (\Uparrow^{\kw{mono}(p)}m)[\Uparrow^{\kw{extend}(p)} \code{\gamma'}]) 
  \end{align*}

  \paragraph{Computation}\label[section]{comp:1}
  \begin{align*}
    \semcmdMA{\code{dcl}(e, m)}(\Sigma', p, \gamma')(\sigma') &= \\
    a &\leftarrow_{\kw{L}} \kw{lift}(\semexpMA{e}(\Sigma', p, \gamma'));\\
    (b, (-, \sigma_1)) &\leftarrow_{\kw{L}} \semcmdMA{m}(A::\Sigma', \kw{mono}(p), \uparrow^{\kw{extend}(p)} \gamma', (a, \sigma')); \\
    &\mstep{1}{\retl{b, \sigma_1}}\\ 
  \end{align*}

  Let $\sigma', \sigma'' : \semsigMA{\Sigma'}$, $b : \semtyMA{B}(\Sigma')$, $c : \Nat$, and 
  suppose that $\semcmdMA{\code{dcl}(e, m)}(\Sigma', p, \gamma')(\sigma') = \mstep{c}{\retl{b, \sigma''}}$. 
  We have to show that if $\mu' \sim_{\Sigma'} \sigma'$, then 
  $(\mu', \code{dcl}((\Uparrow^p e) [\code{\gamma'}], (\Uparrow^{\kw{mono}(p)}m)[\Uparrow^{\kw{extend}(p)} \code{\gamma'}]) ) 
  \Downarrow^{\eta_{\Cl}c}_{\kw{cmd}} (\mu'', \code{ret}(\code{b}))$ for some 
  $\code{b} : \kw{Pg}(\Sigma', B)$ and $\mu'' : \kw{Store}(\Sigma')$ such that 
  $\code{b} \approx_{\Sigma'} b$, and $\mu'' \sim_{\Sigma'} \sigma''$.  

  By $\kw{bind}_{\kw{L}}^{-1}$, $\kw{step}_{\kw{L}}^{-1}$, and $\kw{step}/\kw{inj}$ we have the following:
  \begin{enumerate}
    \item $\semexpMA{e}(\Sigma', p, \gamma') = \mstep{c_1}{\ret{a}}$ for some $a$ and $c_1$. 
    \item $\semcmdMA{m}(A::\Sigma', \kw{mono}(p), \uparrow^{\kw{extend}(p)} \gamma', (a, \sigma')) = \mstep{c_2}{\ret{b, (-, \sigma'')}}$ for some $c_2$. \label{hyp:dcl-m}
    \item $\Cl(c = c_1 + c_2 + 1)$. 
  \end{enumerate}
  Note that \cref{hyp:dcl-m} is type correct because $b : \semtyMA{B}(\Sigma') = \semtyMA{B}(A::\Sigma')$ by \cref{prop:tm-pos}. 
  By induction on $e$, we have that $(\Uparrow^p e) [\code{\gamma'}] \approx^{\Downarrow}_{\Sigma', A} \semexpMA{e}(\Sigma', p, \gamma')$. 
  By definition of the lifting relation, we have that 
  $(e_a, v_a) : (\Uparrow^p e) [\code{\gamma'}] \Downarrow^{\eta_{\Cl} c_1} \code{a}$ for some $\code{a} : \kw{Pg}(\Sigma', A)$ 
  such that $\code{a} \approx_{\Sigma', A} a$. 
  Now suppose that $\mu' \sim_{\Sigma'} \sigma'$. 
  By $\sim.\kw{extend}$, we have a proof of 
  $((\code{a}, v_a)::\mu') \sim_{(A, h_A)::\Sigma'} (a, \sigma')$.  
  By induction on $m$, we know that 
  $(\Uparrow^{\kw{mono}(p)} m)[\Uparrow^{\kw{extend}(p)} \code{\gamma'}] 
  \sim_{(A, h_A)::\Sigma'} \semcmdMA{m}(A::\Sigma', \kw{mono}(p), \uparrow^{\kw{extend}(p)} \gamma')$. 
  Therefore, we know that  
  $((\code{a}, v_a)::\mu', (\Uparrow^{\kw{mono}(p)} m)[\Uparrow^{\kw{extend}(p)} \code{\gamma'}])\Downarrow^{\eta_{\Cl} c_2}
  (-::\mu'', \code{ret}(\code{b}))$ for some $\code{b} : \kw{Pg}(A::\Sigma', B)$ such that 
  $\code{b} \approx_{A::\Sigma'} b$ and $-::\mu'' \sim_{A::\Sigma'} (-, \sigma'')$. 
  Because $B$ is a positive type, by \cref{prop:pos-rel} we have 
  $\kw{coer}(\code{b}) \approx_{\Sigma'} b$.  
  By \cref{prop:eval-dcl}, we know that 
  $(\mu', \code{dcl}((\Uparrow^p e) [\code{\gamma'}], (\Uparrow^{\kw{mono}(p)} m)[\Uparrow^{\kw{extend}(p)} \code{\gamma'}])) 
  \Downarrow^{\eta_{\Cl}(c_1 + c_2 + 1)} (\mu'', \code{ret}(\kw{coer}(\code{b})))$. 

  \subsubsection{\textbf{Case}: set}
  By case we have $e : \TmMA(\Sigma, \Gamma, A)$ and $\Sigma[n] = (A, h_A)$. 
  Suppose $p : \Sigma' \ge \Sigma$ and $\code{\gamma'} \approx_{\Sigma', \Gamma} \gamma'$. 
  We have to show that $(\Uparrow^p \code{set}[n](e)) [\code{\gamma'}] \sim_{\Sigma', A} \semcmdMA{\code{set}[n](e)}(\Sigma', p, \gamma')$.
  \paragraph{Computation}
  \begin{align*}
    (\Uparrow^p \code{set}[n](e)) [\code{\gamma'}] &= \code{set}[\kw{sh}(p, n)](\Uparrow^p e) [\code{\gamma'}]\\
    &=\code{set}[\kw{sh}(p, n)]((\Uparrow^p e) [\code{\gamma'}])  
  \end{align*}

  \paragraph{Computation}
  \begin{align*}
    &\semcmdMA{\code{set}[n](e)}(\Sigma', p, \gamma') =\\
    &\lambda \sigma'.\; 
    a \leftarrow_{\kw{L}} \kw{lift}(\semexpMA{e}(\Sigma', p, \gamma', \sigma')); \step{\retl{\sigma'[\kw{sh}(p, n)], \sigma'[\kw{sh}(p, n) \mapsto a]}}
  \end{align*}
  Suppose $\sigma', \sigma'' : \semsigMA{\Sigma'}$, $a' : \semtyMA{A}(\Sigma')$, $c : \Nat$ and 
  $\semcmdMA{\code{set}[n](e)}(\Sigma', p, \gamma', \sigma') = \mstep{c}{\retl{a', \sigma''}}$. 
  Moreover, suppose that $\mu' \sim_{\Sigma'} \sigma'$. We want to show that 
  $(\mu',\code{set}[\kw{sh}(p, n)]((\Uparrow^p e) [\code{\gamma'}])) \Downarrow^{\eta_{\Cl} c}  (\mu'', \code{ret}(\code{a'}))$ for some
  $\mu'' \sim_{\Sigma'} \sigma''$ and $\code{a'} \approx_{\Sigma', A} a'$.  
  By $\kw{bind}_{\kw{L}}^{-1}$, $\kw{step}_{\kw{L}}^{-1}$, and $\kw{lift}/\kw{inj}$ we have the following: 
  \begin{enumerate}
    \item $\semexpMA{e}(\Sigma', p, \gamma', \sigma') = \mstep{c_1}{\ret{a}}$ for some $a$ and $c_1$.
    \item $a' = \sigma'[\kw{sh}(p, n)]$.
    \item $\sigma'' = \sigma'[\kw{sh}(p, n) \mapsto a]$. 
    \item $\Cl(c = c_1 + 1)$. 
  \end{enumerate}
  By induction on $e$, we have that $(\Uparrow^p e) [\code{\gamma'}] \approx^{\Downarrow}_{\Sigma', A} \semexpMA{e}(\Sigma', p, \gamma')$. 
  Unfolding the lifting relation we have that 
  $(\Uparrow^p e) [\code{\gamma'}] \Downarrow^{\eta_{\Cl} c_1} \code{a}$ for some $\code{a} : \kw{Pg}(\Sigma', A)$ 
  such that $\code{a} \approx_{\Sigma', A} a$. 
  By definition of the operational semantics of \MA{}, we have the following: 
  \[(\mu',\code{set}[\kw{sh}(p, n)](\code{a})) \mapsto_{\kw{cmd}} 
  (\mu'[\kw{sh}(p, n) \mapsto \code{a}], \ret{\mu'[\kw{sh}(p, n)]})\]
  Therefore, we have 
  $(\mu',\code{set}[\kw{sh}(p, n)]((\Uparrow^p e) [\code{\gamma'}])) \Downarrow^{\eta_{\Cl} (c_1 + 1)} 
  (\mu'[\kw{sh}(p, n) \mapsto \code{a}], \code{ret}(\mu'[\kw{sh}(p, n)]))$ by \cref{prop:eval-set}. 
  The result holds by observing that $\mu'[\kw{sh}(p, n)] \approx_{A} \sigma'[\kw{sh}(p, n)]$ and 
  $\mu'[\kw{sh}(p, n) \mapsto \code{a}] \sim_{\Sigma'} \sigma'[\kw{sh}(p, n) \mapsto a]$ since 
  $\mu' \sim_{\Sigma'} \sigma'$ and $\code{a} \approx_A a$. 

%% file: appendix-model.tex
\section{Model construction}\label{appendix:model}

\subsection{Decomposition of cost bounds} 
We use the following proposition of \citet{altenkirch-danielsson-kraus:2017}: 
\begin{proposition}[Inversion]\label{prop:bind-inversion}
  Given $e : A_\bot$ and $f : A \to B_\bot$, 
  if $a \leftarrow_\bot e; f = \eta_\bot(b)$, then there merely exists 
  $a : A$ such that $e = \eta_\bot(a)$ and $f(a) = \eta_\bot(b)$. 
\end{proposition}

We check the axiom $\kw{bind}_{\kw{L}}^{-1}$; the corresponding axiom 
$\kw{step}_{\kw{L}}^{-1}$ may be verified in a similar fashion:
\begin{align*}
  \kw{bind}_{\kw{L}}^{-1} &: \impl{A, B, e, f, c, b} \bindl{e}{f} = \mstep{c}{\retl{b}} \to 
  \lVert \Sigma c_1, c_2 : \mathbb{C}.\; \Sigma a : A. \; e = \mstep{c_1}{\retl{a}} \times\\
  & f(a) = \mstep{c_2}{\retl{b}} \times \Cl(c = c_1 + c_2)\rVert
\end{align*}
Suppose that we have $\bindl{e}{f} = \mstep{c}{\retl{b}}$. 
Computing, this means the following: 
\begin{align*}
  (c_1, a) \leftarrow_\bot e; (c_2, b) \leftarrow_\bot f(a); \eta_\bot(c_1 + c_2, b) = \eta_\bot(\eta_{\Cl} c, b)
\end{align*}
By \cref{prop:bind-inversion}, there merely exists $c_1, c_2 : \Cl\mathbb{C}$, $a : A$, and $b' : B$ 
such that $e = \eta_\bot (c_1, a)$ and $f(a) = \eta_\bot(c_2, b')$ such that 
$\eta_\bot(c_1 + c_2, b') = \eta_\bot(\eta_{\Cl} c, b)$. Because we are proving a proposition, 
we may project out the witness and data of the mere existential. 
First, observe that $b' = b$ and $\eta_{\Cl}c = c_1 + c_2$. 
Therefore, it suffices to show that there are $c_1', c_2' : \mathbb{C}$ such that 
$e = \eta_\bot (\eta_{\Cl} c_1, a)$, $f(a) = \eta_\bot(\eta_{\Cl} c_2, b')$, and 
$\Cl(c = c_1' + c_2')$. Note that if either $c_1$ or $c_2$ is ${\ast}(u)$ for some 
$u : \ExtOpn$, then we may take $c_1' = c_2' = 0$. Otherwise, we have 
$c_1 = \eta_{\Cl} c_1'$ and $c_2 = \eta_{\Cl} c_2'$, and the result holds since 
$\eta_{\Cl}c = c_1 + c_2 = \eta_{\Cl} c_1' + \eta_{\Cl} c_2' = \eta_{\Cl} (c_1' + c_2')$.  

\subsection{Iteration}

In the following we write $L A \coloneqq (\Cl\mathbb{C} \times A)_\bot$ for the 
lift monad. We will define the iteration operator as the fixed-point of the iteration functional:
\begin{align*}
  \kw{ITER} &: \impl{A, B} (A \to L (B + A)) \to (A \to L B) \to (A \to L B)\\
  \kw{ITER}(g, f, a) &= s \leftarrow_L g a; [\eta_L; f] s
\end{align*}

\begin{proposition}[Monotonicity of sequencing \citep{altenkirch-danielsson-kraus:2017}]\label{prop:bind-mono}  
  Given $e, e' : A_\bot$ and $f, f' : A \to B_\bot$, 
  if $e \sqsubseteq e'$ and $f \sqsubseteq f'$ (w.r.t the induced pointwise order), 
  then $a \leftarrow_\bot e; f(a) \sqsubseteq a \leftarrow_\bot e'; f'(a)$. 
\end{proposition}

\begin{lemma}
  Given $g : A \to L(B + A)$, we have that $\kw{ITER}(g)$ is monotone. 
\end{lemma}

\begin{proof}
  Let $f \sqsubseteq f'$ be functions $A \to L B$. We have to show that 
  $s \leftarrow_L g a; [\eta_L; f] s \sqsubseteq s \leftarrow_L g a; [\eta_L; f'] s$ 
  for all $a : A$. By \cref{prop:bind-mono}, it suffices to show 
  that $[\eta_L; f] s \sqsubseteq [\eta_L; f'] s$ for all $s : B + A$. 
  This follows by case analysis on $s$ and the assumption that $f \sqsubseteq f'$.  
\end{proof}

\begin{proposition}[\citet{altenkirch-danielsson-kraus:2017}]\label{prop:bind-sup}
  Given $e : \Nat \to A_\bot$ and $f : A \to B_\bot$, we have that 
  $\bigsqcup (\lambda n.\; \kw{bind}_\bot(e(n), f)) = \kw{bind}_\bot(\bigsqcup e, f)$. 
\end{proposition}

\begin{lemma}
  Given $e : A_\bot$ and $f : \Nat \to A \to B_\bot$, we have that 
  $\bigsqcup (\lambda n.\; \kw{bind}_\bot(e, f(n))) = \kw{bind}_\bot(e, \lambda a.\; \bigsqcup (\lambda n.\; f(n, a)))$. 
\end{lemma}

\begin{proof}
  It suffices to show inclusion for both directions. 
  For the forward direction, fix $n : \Nat$. 
  It suffices to show $\kw{bind}_\bot(e, f(n)) \sqsubseteq \kw{bind}_\bot(e, \lambda a.\; \bigsqcup (\lambda n.\; f(n, a)))$. 
  By \cref{prop:bind-mono}, it suffices to show $f(n, a) \sqsubseteq \bigsqcup (\lambda n.\; f(n, a))$, which 
  clearly holds. 
  For the other direction, we proceed by induction on $e : A_\bot$. 
  We just show the case for $e = \sqcup s$. 
  By \cref{prop:bind-sup}, we have that 
  $\kw{bind}_\bot(\sqcup s, \lambda a.\; \bigsqcup (\lambda n.\; f(n, a))) = 
  \bigsqcup (\lambda m.\; \kw{bind}_\bot(s(m), \lambda a.\; \bigsqcup (\lambda n.\; f(n, a))))$. 
  Fix an arbitrary $m : \Nat$. It suffices to show 
  $\kw{bind}_\bot(s(m), \lambda a.\; \bigsqcup (\lambda n.\; f(n, a))) \sqsubseteq 
  \bigsqcup (\lambda n.\; \kw{bind}_\bot(\sqcup s, f(n)))$. 
  By induction, we have that $\kw{bind}_\bot(s(m), \lambda a.\; \bigsqcup (\lambda n.\; f(n, a)))
  \sqsubseteq \bigsqcup (\lambda n.\; \kw{bind}_\bot(s(m), f(n)))$. 
  Therefore, it suffices to show that 
  $\bigsqcup (\lambda n.\; \kw{bind}_\bot(s(m), f(n))) \sqsubseteq 
  \bigsqcup (\lambda n.\; \kw{bind}_\bot(\sqcup s, f(n)))$, 
  which holds since $s(m) \sqsubseteq \sqcup s$.  
\end{proof}

\begin{corollary}\label{lemma:sup-bind}
  Given $e : L A$ and $f : \Nat \to A \to L B$, we have that 
  $\bigsqcup (\lambda n.\; \kw{bind}_L(e, f(n))) = \kw{bind}_L(e, \lambda a.\; \bigsqcup (\lambda n.\; f(n, a)))$.  
\end{corollary}

\begin{lemma}\label{lemma:iter-cont}
  Given $g : A \to L(B + A)$, we have that $\kw{ITER}(g)$ is $\omega$-continuous. 
\end{lemma}

\begin{proof}
  Suppose $\alpha$ is an $\omega$-chain in $A \to L B$. We have to show the following:  
  \[\kw{ITER}(g)(\bigsqcup \alpha) = \bigsqcup (\lambda n.\; \kw{ITER}(g)(\alpha(n)))\]
  Let $a : A$. We need to show that 
  $s \leftarrow_L g a; [\eta_L; \bigsqcup \alpha] s = \bigsqcup (\lambda n.\; s \leftarrow_L g a; [\eta_L; \alpha(n)] s)$.  
  Computing using \cref{lemma:sup-bind}:
  \begin{align*}
    \bigsqcup (\lambda n.\; s \leftarrow_L g a; [\eta_L; \alpha(n)]) &= 
    s \leftarrow_L g a; \bigsqcup (\lambda n.\; [\eta_L; \alpha(n)] s)
  \end{align*}
  So it suffices to show that 
  $\bigsqcup (\lambda n.\; [\eta_L; \alpha(n)] s) = [\eta_L; \bigsqcup \alpha] s$ for all $s : B + A$. 
  We proceed by cases on $s$. 
  If $s = \inl{b}$, then we compute: 
  \begin{align*}
    \bigsqcup (\lambda n.\; [\eta_L; \alpha(n)] (\inl{b})) &= 
    \bigsqcup (\lambda n.\; \eta_L(b))\\
    &= \eta_L(b)\\ 
    &= [\eta_L; \bigsqcup \alpha] (\inl{b})
  \end{align*}
  Otherwise, $s = \inr{a'}$. Computing: 
  \begin{align*}
    \bigsqcup (\lambda n.\; [\eta_L; \alpha(n)] \inr{a'}) &=  
    \bigsqcup (\lambda n.\; \alpha(n, a'))\\
    &= (\bigsqcup \alpha) a'\\
    &= [\eta_L; \bigsqcup \alpha] (\inr{a'})
  \end{align*}
\end{proof}

By \cref{lemma:iter-cont}, 
we have that $\kw{ITER}(g)$ is an $\omega$-continuous function, and consequently we may take its least 
fixed-point: 
\begin{align*} 
  \kw{iter} &: \impl{A, B} (A \to L(B + A)) \to A \to L B\\ 
  \kw{iter}(g) &= \kw{fix}(\kw{ITER}(g))
\end{align*} 
The unfolding rule of iteration follows from the associated fixed-point equation $\kw{iter}/\kw{unfold}$: 
\begin{align*}
  \kw{iter}(g, a) &= \kw{fix}(\kw{ITER}(g)) a\\
  &= \kw{ITER}(g)(\kw{fix}(\kw{ITER}(g)))(a)\\ 
  &= s \leftarrow_L g a; [\eta_L; \kw{fix}(\kw{ITER}(g))] s
\end{align*}

Lastly, we verify the finiteness axiom for iteration. 
\begin{lemma}\label{lemma:fix-ret}
  Given an $\omega$-continuous function $F : (A \to B_\bot) \to (A \to B_\bot)$, if 
  $\kw{fix}(F)(a) = \eta_\bot(b)$ for some $b : B$, then there merely exists 
  a $k : \Nat$ such that $F^{(k)}(a) = \eta_\bot(b)$, where 
  $F^{(-)} : \Nat \to A \to B_\bot$ is an $\omega$-chain of functions 
  defined by iterating $F$ on the totally undefined function $F^{(0)} = \lambda -.\; \bot$. 
\end{lemma}

\begin{proof}
  By definition, we have $\kw{fix}(F)(a) = \bigsqcup (\lambda k.\; F^{(k)})(a) = \eta_\bot(b)$, 
  and so we have $\eta_\bot(b) \sqsubseteq \bigsqcup (\lambda k.\; F^{(k)}(a))$. 
  By the characterization of $\sqsubseteq$ of \citet{altenkirch-danielsson-kraus:2017}, 
  there merely exists a $k : \Nat$ such that $\eta_\bot(b) \sqsubseteq F^{(k)}(a)$. 
  Conversely, because $\bigsqcup (\lambda k.\; F^{(k)}(a)) \sqsubseteq \eta_\bot(b)$, 
  we have that $F^{(k)}(a) \sqsubseteq \eta_\bot(b)$ as well, and so 
  $F^{(k)}(a) = \eta_\bot(b)$ by anti-symmetry. 
\end{proof}

Now, we apply this to the iteration functional. 
\begin{lemma}  
If $\kw{iter}(g, a) = \mstep{c}{\retl{b}}$, then 
there merely exists a $k : \Nat$ such that 
$\kw{seq}(g, k, a) = \mstep{c}{\retl{\inl{b}}}$. 
\end{lemma}

\begin{proof}
  Suppose that $\kw{iter}(g, a) = \mstep{c}{\retl{b}}$.
  Computing, we have that
  $\kw{fix}(\kw{ITER}(g)) a = \eta_\bot(\eta_{\Cl}(c), b)$, 
  and by \cref{lemma:fix-ret} there merely exists a $k : \Nat$ such that 
  $\kw{ITER}(g)^{(k)}(a) = \eta_\bot(\eta_{\Cl}(c), b)$. 
  We have to show that $\kw{seq}(g, k, a) = \mstep{c}{\retl{\inl{b}}} = \eta_\bot(\eta_{\Cl}(c), \inl{b})$. 
  We prove the following statement by induction: 
  \[
  \Pi k : \Nat.\; \Pi a : A.\; \Pi b : B.\; \Pi c : \Cl\mathbb{C}.\; 
  \kw{ITER}(g)^{(k)}(a) = \eta_\bot(c, b) \to \kw{seq}(g, k, a) = \eta_\bot(c, \inl{b})
  \]
  From which the result follows by applying the fact that 
  $\kw{ITER}(g)^{(k)}(a) = \eta_\bot(\eta_{\Cl}(c), b)$.  
  If $k = 0$, then we have $\kw{ITER}(g)^{(k)}(a) = \bot = \eta_\bot(\eta_{\Cl}(c), b)$, which is 
  a contradiction. 
  Otherwise, we have $k = k' + 1$ for some $k' : \Nat$. 
  Computing: 
  \begin{align*}
    (\kw{ITER}(g))^{(k' + 1)}(a) &= 
    \kw{ITER}(g)(\kw{ITER}(g)^{(k')})(a)\\
    &= s \leftarrow_L g a; [\eta_L; \kw{ITER}(g)^{(k')}]s\\
    &= (c_1, s) \leftarrow_\bot g a; (c_2, b) \leftarrow_\bot [\eta_T; \kw{ITER}(g)^{(k')}]s; \eta_\bot(c_1 + c_2, b)
  \end{align*} 
  By \cref{prop:bind-inversion}, we have the following:
  \begin{enumerate}
    \item $g a = \eta_\bot(c_1, s)$ for some $c_1 : \Cl\mathbb{C}$ and $s : B + A$. 
    \item $[\eta_L; \kw{ITER}(g)^{(k')}]s = \eta_\bot(c_2, b')$ for some $c_2 : \Cl\mathbb{C}$ and $b' : B$. 
    \item $\eta_\bot(c_1 + c_2, b) = \eta_\bot(\eta_{\Cl}(c), b')$. 
  \end{enumerate}
  From the last line, we know that $b = b'$ and $c_1 + c_2 = \eta_{\Cl}(c)$. 
  Computing the sequence: 
  \begin{align*}
    \kw{seq}(g, k'+1, a) &= s \leftarrow_L g a; [\eta_L \circ \kw{inl}; \kw{seq}(g, k')]s  \\
    &= (c_1, s) \leftarrow_\bot g a; (c_2, r) \leftarrow_\bot [\eta_L \circ \kw{inl};  \kw{seq}(g, k')] s; \eta_\bot(c_1 + c_2, r)\\
    &= (c_2, r) \leftarrow_\bot [\eta_L \circ \kw{inl}; \kw{seq}(g, k')] s; \eta_\bot(c_1 + c_2, r)
  \end{align*}
  We want to show that $(c_2, r) \leftarrow_\bot [\eta_L \circ \kw{inl}; \kw{seq}(g, k')]s; 
  \eta_\bot(c_1 + c_2, r) = \eta_\bot(\eta_{\Cl}(c), \inl{b})$. 
  Proceed by cases on $s : B + A$. 
  If $s = \inl{b}$, then we have 
  \begin{align*}
    &(c_2, r) \leftarrow_\bot [\eta_L \circ \kw{inl}; \kw{seq}(g, k')](\inl{b}); 
  \eta_\bot(c_1 + c_2, r)\\
  &=  (c_2, r) \leftarrow_\bot \eta_L \circ \kw{inl}(b); \eta_\bot(c_1 + c_2, r) \\
  &= (c_2, r) \leftarrow_\bot \eta_\bot(\eta_{\Cl} 0, \kw{inl}(b)); \eta_\bot(c_1 + c_2, r) \\
  &= \eta_\bot(c_1, \kw{inl}(b)) 
  \end{align*}
  So it suffices to show that $c_1 = \eta_{\Cl} c$. From above, 
  we know that $[\eta_L; \kw{ITER}(g)^{(k')}]s = \eta_L(b) = \eta_\bot(\eta_{\Cl}0, b) = \eta_\bot(c_2, b)$, 
  and so $c_2 = \eta_{\Cl 0}$, from which the result follows since $c_1 + c_2 = c_1 = \eta_{\Cl} c$.  
  Otherwise, $s = \inr{a'}$ and we have 
  $[\eta_L; \kw{ITER}(g)^{(k')}](\inr{a'}) = \kw{ITER}(g)^{(k')}(a') = \eta_\bot(c_2, b)$. 
  By induction hypothesis, we have that 
  $\kw{seq}(g, k', a') = \eta_\bot(c_2, \inl{b})$. 
  Now compute:
  \begin{align*}
    &(c_2, r) \leftarrow_\bot [\eta_L \circ \kw{inl}; \kw{seq}(g, k')](\inr{a'}); 
  \eta_\bot(c_1 + c_2, r) \\
  &= (c_2, r) \leftarrow_\bot \kw{seq}(g, k', a'); \eta_\bot(c_1 + c_2, r) \tag{Induction}\\ 
  &= (c_2, r) \leftarrow_\bot \eta_\bot(c_2, \inl{b}); \eta_\bot(c_1 + c_2, r) \\
  &= \eta_\bot(c_1 + c_2, \inl{b})
  \end{align*}
\end{proof}

%% file: main.bbl

\begin{thebibliography}{31}


\ifx \showCODEN    \undefined \def \showCODEN     #1{\unskip}     \fi
\ifx \showDOI      \undefined \def \showDOI       #1{#1}\fi
\ifx \showISBNx    \undefined \def \showISBNx     #1{\unskip}     \fi
\ifx \showISBNxiii \undefined \def \showISBNxiii  #1{\unskip}     \fi
\ifx \showISSN     \undefined \def \showISSN      #1{\unskip}     \fi
\ifx \showLCCN     \undefined \def \showLCCN      #1{\unskip}     \fi
\ifx \shownote     \undefined \def \shownote      #1{#1}          \fi
\ifx \showarticletitle \undefined \def \showarticletitle #1{#1}   \fi
\ifx \showURL      \undefined \def \showURL       {\relax}        \fi
\providecommand\bibfield[2]{#2}
\providecommand\bibinfo[2]{#2}
\providecommand\natexlab[1]{#1}
\providecommand\showeprint[2][]{arXiv:#2}

\bibitem[\protect\citeauthoryear{Ahmed}{Ahmed}{2015}]%
        {ahmed15:snapl}
\bibfield{author}{\bibinfo{person}{Amal Ahmed}.}
  \bibinfo{year}{2015}\natexlab{}.
\newblock \showarticletitle{{Verified Compilers for a Multi-Language World}}.
  In \bibinfo{booktitle}{\emph{1st Summit on Advances in Programming Languages
  (SNAPL 2015)}} \emph{(\bibinfo{series}{Leibniz International Proceedings in
  Informatics (LIPIcs)}, Vol.~\bibinfo{volume}{32})},
  \bibfield{editor}{\bibinfo{person}{Thomas Ball}, \bibinfo{person}{Rastislav
  Bodik}, \bibinfo{person}{Shriram Krishnamurthi}, \bibinfo{person}{Benjamin~S.
  Lerner}, {and} \bibinfo{person}{Greg Morrisett}} (Eds.).
  \bibinfo{address}{Asilomar, California}, \bibinfo{pages}{15--31}.
\newblock


\bibitem[\protect\citeauthoryear{Altenkirch, Danielsson, and Kraus}{Altenkirch
  et~al\mbox{.}}{2017}]%
        {altenkirch-danielsson-kraus:2017}
\bibfield{author}{\bibinfo{person}{Thorsten Altenkirch},
  \bibinfo{person}{Nils~Anders Danielsson}, {and} \bibinfo{person}{Nicolai
  Kraus}.} \bibinfo{year}{2017}\natexlab{}.
\newblock \showarticletitle{Partiality, Revisited: The Partiality Monad as a
  Quotient Inductive-Inductive Type}. In \bibinfo{booktitle}{\emph{Foundations
  of Software Science and Computation Structures}},
  \bibfield{editor}{\bibinfo{person}{Javier Esparza} {and}
  \bibinfo{person}{Andrzej~S. Murawski}} (Eds.). \bibinfo{publisher}{Springer
  Berlin Heidelberg}, \bibinfo{address}{Berlin, Heidelberg},
  \bibinfo{pages}{534--549}.
\newblock
\showISBNx{978-3-662-54458-7}
\urldef\tempurl%
\url{https://doi.org/10.1007/978-3-662-54458-7_31}
\showDOI{\tempurl}
\showeprint[arXiv]{1610.09254}~[cs.LO]


\bibitem[\protect\citeauthoryear{Altenkirch, Ghani, Hancock, McBride, and
  Morris}{Altenkirch et~al\mbox{.}}{2015}]%
        {altenkirch-ghani-hancock-mcbride-morris:2015}
\bibfield{author}{\bibinfo{person}{Thorsten Altenkirch}, \bibinfo{person}{Neil
  Ghani}, \bibinfo{person}{Peter Hancock}, \bibinfo{person}{Conor McBride},
  {and} \bibinfo{person}{Peter Morris}.} \bibinfo{year}{2015}\natexlab{}.
\newblock \showarticletitle{Indexed containers}.
\newblock \bibinfo{journal}{\emph{Journal of Functional Programming}}
  \bibinfo{volume}{25} (\bibinfo{year}{2015}).
\newblock


\bibitem[\protect\citeauthoryear{Benton and Hur}{Benton and Hur}{2010}]%
        {benton2010realizability}
\bibfield{author}{\bibinfo{person}{Nick Benton} {and}
  \bibinfo{person}{Chung-Kil Hur}.} \bibinfo{year}{2010}\natexlab{}.
\newblock \bibinfo{booktitle}{\emph{Realizability and Compositional Compiler
  Correctness for a Polymorphic Language}}.
\newblock \bibinfo{type}{{T}echnical {R}eport} MSR-TR-2010-62.
  \bibinfo{institution}{Microsoft Research}.
\newblock
\urldef\tempurl%
\url{https://www.microsoft.com/en-us/research/publication/realizability-and-compositional-compiler-correctness-for-a-polymorphic-language/}
\showURL{%
\tempurl}


\bibitem[\protect\citeauthoryear{Benton, Hur, Kennedy, and McBride}{Benton
  et~al\mbox{.}}{2012}]%
        {benton-hur-kennedy-mcbride:2012}
\bibfield{author}{\bibinfo{person}{Nick Benton}, \bibinfo{person}{Chung-Kil
  Hur}, \bibinfo{person}{Andrew~J. Kennedy}, {and} \bibinfo{person}{Conor
  McBride}.} \bibinfo{year}{2012}\natexlab{}.
\newblock \showarticletitle{Strongly Typed Term Representations in Coq}.
\newblock \bibinfo{journal}{\emph{Journal of Automated Reasoning}}
  \bibinfo{volume}{49}, \bibinfo{number}{2} (\bibinfo{year}{2012}),
  \bibinfo{pages}{141--159}.
\newblock
\showISBNx{1573-0670}
\urldef\tempurl%
\url{https://doi.org/10.1007/s10817-011-9219-0}
\showDOI{\tempurl}


\bibitem[\protect\citeauthoryear{Birkedal, M{\o{}}gelberg, Schwinghammer, and
  St\o{}vring}{Birkedal et~al\mbox{.}}{2011}]%
        {bmss:2011}
\bibfield{author}{\bibinfo{person}{Lars Birkedal},
  \bibinfo{person}{Rasmus~Ejlers M{\o{}}gelberg}, \bibinfo{person}{Jan
  Schwinghammer}, {and} \bibinfo{person}{Kristian St\o{}vring}.}
  \bibinfo{year}{2011}\natexlab{}.
\newblock \showarticletitle{First Steps in Synthetic Guarded Domain Theory:
  Step-Indexing in the Topos of Trees}. In
  \bibinfo{booktitle}{\emph{Proceedings of the 2011 IEEE 26th Annual Symposium
  on Logic in Computer Science}}. \bibinfo{publisher}{IEEE Computer Society},
  \bibinfo{address}{Washington, DC, USA}, \bibinfo{pages}{55--64}.
\newblock
\showISBNx{978-0-7695-4412-0}
\urldef\tempurl%
\url{https://doi.org/10.1109/LICS.2011.16}
\showDOI{\tempurl}
\showeprint[arXiv]{1208.3596}~[cs.LO]


\bibitem[\protect\citeauthoryear{Danner, Licata, and Ramyaa}{Danner
  et~al\mbox{.}}{2015}]%
        {danner-licata-ramyaa:2015}
\bibfield{author}{\bibinfo{person}{Norman Danner}, \bibinfo{person}{Daniel~R.
  Licata}, {and} \bibinfo{person}{Ramyaa}.} \bibinfo{year}{2015}\natexlab{}.
\newblock \showarticletitle{Denotational cost semantics for functional
  languages with inductive types}. In \bibinfo{booktitle}{\emph{Proceedings of
  the 20th {ACM} {SIGPLAN} International Conference on Functional Programming,
  {ICFP} 2015, Vancouver, BC, Canada, September 1-3, 2015}},
  \bibfield{editor}{\bibinfo{person}{Kathleen Fisher} {and}
  \bibinfo{person}{John~H. Reppy}} (Eds.). \bibinfo{publisher}{Association for
  Computing Machinery}, \bibinfo{pages}{140--151}.
\newblock
\urldef\tempurl%
\url{https://doi.org/10.1145/2784731.2784749}
\showDOI{\tempurl}


\bibitem[\protect\citeauthoryear{Fiore and Plotkin}{Fiore and Plotkin}{1996}]%
        {fiore-plotkin:1996}
\bibfield{author}{\bibinfo{person}{Marcelo~P. Fiore} {and}
  \bibinfo{person}{Gordon~D. Plotkin}.} \bibinfo{year}{1996}\natexlab{}.
\newblock \showarticletitle{An Extension of Models of Axiomatic Domain Theory
  to Models of Synthetic Domain Theory}. In \bibinfo{booktitle}{\emph{Computer
  Science Logic, 10th International Workshop, {CSL} '96, Annual Conference of
  the EACSL, Utrecht, The Netherlands, September 21-27, 1996, Selected Papers}}
  \emph{(\bibinfo{series}{Lecture Notes in Computer Science},
  Vol.~\bibinfo{volume}{1258})}, \bibfield{editor}{\bibinfo{person}{Dirk van
  Dalen} {and} \bibinfo{person}{Marc Bezem}} (Eds.).
  \bibinfo{publisher}{Springer}, \bibinfo{pages}{129--149}.
\newblock
\urldef\tempurl%
\url{https://doi.org/10.1007/3-540-63172-0\_36}
\showDOI{\tempurl}


\bibitem[\protect\citeauthoryear{Fiore and Rosolini}{Fiore and
  Rosolini}{1997}]%
        {fiore-rosolini:1997}
\bibfield{author}{\bibinfo{person}{Marcelo~P. Fiore} {and}
  \bibinfo{person}{Giuseppe Rosolini}.} \bibinfo{year}{1997}\natexlab{}.
\newblock \showarticletitle{Two models of synthetic domain theory}.
\newblock \bibinfo{journal}{\emph{Journal of Pure and Applied Algebra}}
  \bibinfo{volume}{116}, \bibinfo{number}{1} (\bibinfo{year}{1997}),
  \bibinfo{pages}{151--162}.
\newblock
\showISSN{0022-4049}
\urldef\tempurl%
\url{https://doi.org/10.1016/S0022-4049(96)00164-8}
\showDOI{\tempurl}


\bibitem[\protect\citeauthoryear{Harper}{Harper}{2012}]%
        {harper:2012:pfpl}
\bibfield{author}{\bibinfo{person}{Robert Harper}.}
  \bibinfo{year}{2012}\natexlab{}.
\newblock \bibinfo{booktitle}{\emph{Practical Foundations for Programming
  Languages} (\bibinfo{edition}{first} ed.)}.
\newblock \bibinfo{publisher}{Cambridge University Press},
  \bibinfo{address}{New York, NY, USA}.
\newblock


\bibitem[\protect\citeauthoryear{Hyland}{Hyland}{1991}]%
        {hyland:1991}
\bibfield{author}{\bibinfo{person}{J.~M.~E. Hyland}.}
  \bibinfo{year}{1991}\natexlab{}.
\newblock \showarticletitle{First steps in synthetic domain theory}. In
  \bibinfo{booktitle}{\emph{Category Theory}},
  \bibfield{editor}{\bibinfo{person}{Aurelio Carboni},
  \bibinfo{person}{Maria~Cristina Pedicchio}, {and} \bibinfo{person}{Guiseppe
  Rosolini}} (Eds.). \bibinfo{publisher}{Springer Berlin Heidelberg},
  \bibinfo{address}{Berlin, Heidelberg}, \bibinfo{pages}{131--156}.
\newblock
\showISBNx{978-3-540-46435-8}


\bibitem[\protect\citeauthoryear{Jung, Fiore, Moggi, O’Hearn, Riecke,
  Rosolini, and Stark}{Jung et~al\mbox{.}}{1996}]%
        {Jung1996DomainsAD}
\bibfield{author}{\bibinfo{person}{Achim Jung}, \bibinfo{person}{Marcelo
  Fiore}, \bibinfo{person}{Eugenio Moggi}, \bibinfo{person}{Peter~W O’Hearn},
  \bibinfo{person}{Jon~G Riecke}, \bibinfo{person}{Giuseppe Rosolini}, {and}
  \bibinfo{person}{Ian Stark}.} \bibinfo{year}{1996}\natexlab{}.
\newblock \showarticletitle{Domains and denotational semantics: History,
  accomplishments and open problems}.
\newblock \bibinfo{journal}{\emph{SCHOOL OF COMPUTER SCIENCE RESEARCH
  REPORTS-UNIVERSITY OF BIRMINGHAM CSR}} (\bibinfo{year}{1996}).
\newblock


\bibitem[\protect\citeauthoryear{Kavvos, Morehouse, Licata, and Danner}{Kavvos
  et~al\mbox{.}}{2019}]%
        {kavvos-morehouse-licata-danner:2019}
\bibfield{author}{\bibinfo{person}{G.~A. Kavvos}, \bibinfo{person}{Edward
  Morehouse}, \bibinfo{person}{Daniel~R. Licata}, {and} \bibinfo{person}{Norman
  Danner}.} \bibinfo{year}{2019}\natexlab{}.
\newblock \showarticletitle{Recurrence Extraction for Functional Programs
  through Call-by-Push-Value}.
\newblock \bibinfo{journal}{\emph{Proceedings of the ACM on Programming
  Languages}} \bibinfo{volume}{4}, \bibinfo{number}{POPL} (\bibinfo{date}{Dec.}
  \bibinfo{year}{2019}).
\newblock
\urldef\tempurl%
\url{https://doi.org/10.1145/3371083}
\showDOI{\tempurl}


\bibitem[\protect\citeauthoryear{Mates, Perconti, and Ahmed}{Mates
  et~al\mbox{.}}{2019}]%
        {ahmed19:refcc}
\bibfield{author}{\bibinfo{person}{Phillip Mates}, \bibinfo{person}{Jamie
  Perconti}, {and} \bibinfo{person}{Amal Ahmed}.}
  \bibinfo{year}{2019}\natexlab{}.
\newblock \showarticletitle{Under Control: Compositionally Correct Closure
  Conversion with Mutable State}. In \bibinfo{booktitle}{\emph{ACM Conference
  on Principles and Practice of Declarative Programming (PPDP)}}.
  \bibinfo{address}{Porto, Portugal}.
\newblock


\bibitem[\protect\citeauthoryear{M\o{}gelberg and Paviotti}{M\o{}gelberg and
  Paviotti}{2016}]%
        {mogelberg-paviotti:2016}
\bibfield{author}{\bibinfo{person}{Rasmus~Ejlers M\o{}gelberg} {and}
  \bibinfo{person}{Marco Paviotti}.} \bibinfo{year}{2016}\natexlab{}.
\newblock \showarticletitle{Denotational Semantics of Recursive Types in
  Synthetic Guarded Domain Theory}. In \bibinfo{booktitle}{\emph{Proceedings of
  the 31st Annual ACM/IEEE Symposium on Logic in Computer Science}}.
  \bibinfo{publisher}{Association for Computing Machinery},
  \bibinfo{address}{New York, NY, USA}, \bibinfo{pages}{317--326}.
\newblock
\showISBNx{978-1-4503-4391-6}
\urldef\tempurl%
\url{https://doi.org/10.1145/2933575.2934516}
\showDOI{\tempurl}


\bibitem[\protect\citeauthoryear{Niu, Sterling, Grodin, and Harper}{Niu
  et~al\mbox{.}}{2022}]%
        {niu-sterling-grodin-harper:2022}
\bibfield{author}{\bibinfo{person}{Yue Niu}, \bibinfo{person}{Jonathan
  Sterling}, \bibinfo{person}{Harrison Grodin}, {and} \bibinfo{person}{Robert
  Harper}.} \bibinfo{year}{2022}\natexlab{}.
\newblock \showarticletitle{A Cost-Aware Logical Framework}.
\newblock \bibinfo{journal}{\emph{Proceedings of the ACM on Programming
  Languages}} \bibinfo{volume}{6}, \bibinfo{number}{POPL} (\bibinfo{date}{Jan.}
  \bibinfo{year}{2022}).
\newblock
\urldef\tempurl%
\url{https://doi.org/10.1145/3498670}
\showDOI{\tempurl}
\showeprint[arXiv]{2107.04663}~[cs.PL]


\bibitem[\protect\citeauthoryear{O’Hearn and Tennent}{O’Hearn and
  Tennent}{1997a}]%
        {ohearn:1997}
\bibfield{editor}{\bibinfo{person}{Peter~W. O’Hearn} {and}
  \bibinfo{person}{Robert~D. Tennent}} (Eds.).
  \bibinfo{year}{1997}\natexlab{a}.
\newblock \bibinfo{booktitle}{\emph{Algol-like {Languages}}}.
  Vol.~\bibinfo{volume}{1}.
\newblock \bibinfo{publisher}{Birkhäuser Boston}, \bibinfo{address}{Boston,
  MA}.
\newblock
\urldef\tempurl%
\url{https://doi.org/10.1007/978-1-4612-4118-8}
\showDOI{\tempurl}


\bibitem[\protect\citeauthoryear{O’Hearn and Tennent}{O’Hearn and
  Tennent}{1997b}]%
        {ohearn:1997-2}
\bibfield{editor}{\bibinfo{person}{Peter~W. O’Hearn} {and}
  \bibinfo{person}{Robert~D. Tennent}} (Eds.).
  \bibinfo{year}{1997}\natexlab{b}.
\newblock \bibinfo{booktitle}{\emph{Algol-like {Languages}}}.
  Vol.~\bibinfo{volume}{2}.
\newblock \bibinfo{publisher}{Birkhäuser Boston}, \bibinfo{address}{Boston,
  MA}.
\newblock
\urldef\tempurl%
\url{https://doi.org/10.1007/978-1-4757-3851-3}
\showDOI{\tempurl}


\bibitem[\protect\citeauthoryear{Patterson and Ahmed}{Patterson and
  Ahmed}{2019}]%
        {patterson19:ccc}
\bibfield{author}{\bibinfo{person}{Daniel Patterson} {and}
  \bibinfo{person}{Amal Ahmed}.} \bibinfo{year}{2019}\natexlab{}.
\newblock \showarticletitle{The Next 700 Compiler Correctness Theorems
  (Functional Pearl)}. In \bibinfo{booktitle}{\emph{{I}nternational
  {C}onference on {F}unctional {P}rogramming ({ICFP}), Berlin, Germany}}.
  \bibinfo{publisher}{ACM Press}, \bibinfo{address}{Berlin, Germany}.
\newblock


\bibitem[\protect\citeauthoryear{Paviotti, M\o{}gelberg, and Birkedal}{Paviotti
  et~al\mbox{.}}{2015}]%
        {paviotti-mogelberg-birkedal:2015}
\bibfield{author}{\bibinfo{person}{Marco Paviotti},
  \bibinfo{person}{Rasmus~Ejlers M\o{}gelberg}, {and} \bibinfo{person}{Lars
  Birkedal}.} \bibinfo{year}{2015}\natexlab{}.
\newblock \showarticletitle{A Model of {PCF} in {Guarded Type Theory}}.
\newblock \bibinfo{journal}{\emph{Electronic Notes in Theoretical Computer
  Science}} \bibinfo{volume}{319}, \bibinfo{number}{Supplement C}
  (\bibinfo{year}{2015}), \bibinfo{pages}{333--349}.
\newblock
\showISSN{1571-0661}
\urldef\tempurl%
\url{https://doi.org/10.1016/j.entcs.2015.12.020}
\showDOI{\tempurl}
\newblock
\shownote{The 31st Conference on the Mathematical Foundations of Programming
  Semantics (MFPS XXXI).}


\bibitem[\protect\citeauthoryear{P\'{e}drot and Tabareau}{P\'{e}drot and
  Tabareau}{2019}]%
        {pedrot-tabareau:2020}
\bibfield{author}{\bibinfo{person}{Pierre-Marie P\'{e}drot} {and}
  \bibinfo{person}{Nicolas Tabareau}.} \bibinfo{year}{2019}\natexlab{}.
\newblock \showarticletitle{The Fire Triangle: How to Mix Substitution,
  Dependent Elimination, and Effects}.
\newblock \bibinfo{journal}{\emph{Proceedings of the ACM on Programming
  Languages}} \bibinfo{volume}{4}, \bibinfo{number}{POPL} (\bibinfo{date}{Dec.}
  \bibinfo{year}{2019}).
\newblock
\urldef\tempurl%
\url{https://doi.org/10.1145/3371126}
\showDOI{\tempurl}


\bibitem[\protect\citeauthoryear{Perconti and Ahmed}{Perconti and
  Ahmed}{2014}]%
        {perconti14:fca}
\bibfield{author}{\bibinfo{person}{James~T. Perconti} {and}
  \bibinfo{person}{Amal Ahmed}.} \bibinfo{year}{2014}\natexlab{}.
\newblock \showarticletitle{Verifying an Open Compiler Using Multi-Language
  Semantics}. In \bibinfo{booktitle}{\emph{European Symposium on Programming
  (ESOP)}}. \bibinfo{address}{Grenoble, France}.
\newblock


\bibitem[\protect\citeauthoryear{Plotkin}{Plotkin}{1977}]%
        {PLOTKIN1977223}
\bibfield{author}{\bibinfo{person}{G.D. Plotkin}.}
  \bibinfo{year}{1977}\natexlab{}.
\newblock \showarticletitle{LCF considered as a programming language}.
\newblock \bibinfo{journal}{\emph{Theoretical Computer Science}}
  \bibinfo{volume}{5}, \bibinfo{number}{3} (\bibinfo{year}{1977}),
  \bibinfo{pages}{223--255}.
\newblock
\showISSN{0304-3975}
\urldef\tempurl%
\url{https://doi.org/10.1016/0304-3975(77)90044-5}
\showDOI{\tempurl}


\bibitem[\protect\citeauthoryear{Reus and Streicher}{Reus and
  Streicher}{1999}]%
        {reus-streicher:1999}
\bibfield{author}{\bibinfo{person}{Bernhard Reus} {and} \bibinfo{person}{Thomas
  Streicher}.} \bibinfo{year}{1999}\natexlab{}.
\newblock \showarticletitle{General synthetic domain theory --- a logical
  approach}.
\newblock \bibinfo{journal}{\emph{Mathematical Structures in Computer Science}}
  \bibinfo{volume}{9}, \bibinfo{number}{2} (\bibinfo{year}{1999}),
  \bibinfo{pages}{177--223}.
\newblock
\urldef\tempurl%
\url{https://doi.org/10.1017/S096012959900273X}
\showDOI{\tempurl}


\bibitem[\protect\citeauthoryear{Reynolds}{Reynolds}{1981}]%
        {reynolds:1981}
\bibfield{author}{\bibinfo{person}{John~C. Reynolds}.}
  \bibinfo{year}{1981}\natexlab{}.
\newblock \showarticletitle{The {Essence} of {ALGOL}}. In
  \bibinfo{booktitle}{\emph{Algorithmic {Languages}: {Proceedings} of the
  {International} {Symposium} on {Algorithmic} {Languages}}},
  \bibfield{editor}{\bibinfo{person}{J.~W. de~Bakker} {and}
  \bibinfo{person}{J.~C. van Vliet}} (Eds.).
  \bibinfo{publisher}{North-Holland}, \bibinfo{address}{Amsterdam},
  \bibinfo{pages}{345--372}.
\newblock


\bibitem[\protect\citeauthoryear{Scott and Strachey}{Scott and
  Strachey}{1971}]%
        {scott-strachey:1971}
\bibfield{author}{\bibinfo{person}{Dana Scott} {and} \bibinfo{person}{C.
  Strachey}.} \bibinfo{year}{1971}\natexlab{}.
\newblock \showarticletitle{Towards a Mathematical Semantics for Computer
  Languages}.
\newblock \bibinfo{journal}{\emph{Proceedings of the Symposium on Computers and
  Automata}}  \bibinfo{volume}{21} (\bibinfo{date}{01} \bibinfo{year}{1971}).
\newblock


\bibitem[\protect\citeauthoryear{Scott}{Scott}{1970}]%
        {scott:1970:outline}
\bibfield{author}{\bibinfo{person}{Dana~S. Scott}.}
  \bibinfo{year}{1970}\natexlab{}.
\newblock \bibinfo{booktitle}{\emph{Outline of a Mathematical Theory of
  Computation}}.
\newblock \bibinfo{type}{{T}echnical {R}eport} PRG02.
  \bibinfo{institution}{Oxford University Computer Laboratory}.
  \bibinfo{pages}{30} pages.
\newblock


\bibitem[\protect\citeauthoryear{Scott}{Scott}{1982}]%
        {scott:1982}
\bibfield{author}{\bibinfo{person}{Dana~S. Scott}.}
  \bibinfo{year}{1982}\natexlab{}.
\newblock \showarticletitle{Domains for denotational semantics}. In
  \bibinfo{booktitle}{\emph{Automata, Languages and Programming}},
  \bibfield{editor}{\bibinfo{person}{Mogens Nielsen} {and}
  \bibinfo{person}{Erik~Meineche Schmidt}} (Eds.). \bibinfo{publisher}{Springer
  Berlin Heidelberg}, \bibinfo{address}{Berlin, Heidelberg},
  \bibinfo{pages}{577--610}.
\newblock
\showISBNx{978-3-540-39308-5}


\bibitem[\protect\citeauthoryear{Sterling and Harper}{Sterling and
  Harper}{2021}]%
        {sterling-harper:2021}
\bibfield{author}{\bibinfo{person}{Jonathan Sterling} {and}
  \bibinfo{person}{Robert Harper}.} \bibinfo{year}{2021}\natexlab{}.
\newblock \showarticletitle{Logical Relations as Types: Proof-Relevant
  Parametricity for Program Modules}.
\newblock \bibinfo{journal}{\emph{J. ACM}} \bibinfo{volume}{68},
  \bibinfo{number}{6} (\bibinfo{date}{Oct.} \bibinfo{year}{2021}).
\newblock
\showISSN{0004-5411}
\urldef\tempurl%
\url{https://doi.org/10.1145/3474834}
\showDOI{\tempurl}
\showeprint[arXiv]{2010.08599}~[cs.PL]


\bibitem[\protect\citeauthoryear{Sterling and Harper}{Sterling and
  Harper}{2022}]%
        {sterling-harper:2022}
\bibfield{author}{\bibinfo{person}{Jonathan Sterling} {and}
  \bibinfo{person}{Robert Harper}.} \bibinfo{year}{2022}\natexlab{}.
\newblock \showarticletitle{Sheaf semantics of termination-insensitive
  noninterference}. In \bibinfo{booktitle}{\emph{7th International Conference
  on Formal Structures for Computation and Deduction (FSCD 2022)}}
  \emph{(\bibinfo{series}{Leibniz International Proceedings in Informatics
  (LIPIcs)}, Vol.~\bibinfo{volume}{228})},
  \bibfield{editor}{\bibinfo{person}{Amy Felty}} (Ed.).
  \bibinfo{publisher}{Schloss Dagstuhl--Leibniz-Zentrum fuer Informatik},
  \bibinfo{address}{Dagstuhl, Germany}.
\newblock
\urldef\tempurl%
\url{https://doi.org/10.4230/LIPIcs.FSCD.2022.15}
\showDOI{\tempurl}
\showeprint[arXiv]{2204.09421}~[cs.PL]


\bibitem[\protect\citeauthoryear{{Univalent Foundations Program}}{{Univalent
  Foundations Program}}{2013}]%
        {hottbook}
\bibfield{author}{\bibinfo{person}{The {Univalent Foundations Program}}.}
  \bibinfo{year}{2013}\natexlab{}.
\newblock \bibinfo{booktitle}{\emph{Homotopy Type Theory: Univalent Foundations
  of Mathematics}}.
\newblock \bibinfo{publisher}{\url{https://homotopytypetheory.org/book}},
  \bibinfo{address}{Institute for Advanced Study}.
\newblock


\end{thebibliography}
